\documentclass[10pt]{article}

\usepackage{geometry}
\usepackage{lmodern}
\usepackage[T1]{fontenc}
\usepackage[mathletters]{ucs} 
\usepackage[utf8]{inputenc}

\usepackage{graphicx}
\graphicspath{{images/}}

\usepackage{amsmath, amsthm}
\usepackage{amssymb, amsbsy}
\usepackage[dvipsnames]{xcolor}
\usepackage{hyperref}
\hypersetup{
colorlinks=true,
linkcolor=Brown,
citecolor=OliveGreen,
urlcolor=RoyalBlue,
}

\usepackage{floatflt, wrapfig}
\usepackage{xcolor}

\usepackage{graphicx}
\graphicspath{ {images/} }
\usepackage{epstopdf}
\usepackage{color, import}
\usepackage{sectsty}
\usepackage{enumitem}
\usepackage{lineno}

\usepackage{cite}
\usepackage[edges]{forest}

\usepackage[linesnumbered,vlined,ruled]{algorithm2e}
\SetAlFnt{\footnotesize}
\SetAlCapFnt{\small}
\SetAlCapNameFnt{\small}

\usepackage{setspace}
\usepackage[list=true]{subcaption}
\usepackage[title]{appendix}
\usepackage{multirow}
\usepackage{array}

\usepackage{etoolbox}
\makeatletter
\patchcmd{\@maketitle}{\begin{center}}{\begin{flushleft}}{}{}
\patchcmd{\@maketitle}{\begin{tabular}[t]{c}}{\begin{tabular}[t]{@{}l}}{}{}
\patchcmd{\@maketitle}{\end{center}}{\end{flushleft}}{}{}
\makeatother

\let\oldnl\nl
\newcommand{\nonl}{\renewcommand{\nl}{\let\nl\oldnl}}

\renewenvironment{abstract}
{\small\section*
{\bfseries\noindent{\raisebox{-.15\baselineskip}{\normalsize\abstractname}}\hrulefill} 
}

\newtheorem{theorem}{Theorem}
\newtheorem{lemma}{Lemma}

\author{hkhj\\hjk}
\author{hjjkhj\\hjk}

\begin{document}
 \pagenumbering{gobble}

\title{Arbitrary Pattern Formation by Asynchronous Opaque Robots with Lights \footnote{A preliminary version \cite{Bose2019} of this paper appeared in the 26th International Colloquium on Structural Information and Communication Complexity (SIROCCO 2019), July 1-4, 2019, L’Aquila, Italy.}}
    \author{    
    Kaustav Bose\\
    \small \emph{Department of Mathematics, Jadavpur University, India}\\
    \small \emph{kaustavbose.rs@jadavpuruniversity.in}\\\\
    Manash Kumar Kundu\\
    \small \emph{Gayeshpur Government Polytechnic, Kalyani, India}\\
    \small \emph{manashkrkundu.rs@jadavpuruniversity.in}\\\\
    Ranendu Adhikary\\
    \small \emph{Department of Mathematics, Jadavpur University, India}\\
    \small \emph{ranenduadhikary.rs@jadavpuruniversity.in}\\\\
    Buddhadeb Sau\\
    \small \emph{Department of Mathematics, Jadavpur University, India}\\
    \small \emph{buddhadeb.sau@jadavpuruniversity.in}\\\\
    }
    \date{}
    
  \maketitle

\begin{abstract}


The \textsc{Arbitrary Pattern Formation} problem asks for a distributed algorithm that moves a set of autonomous mobile robots to form any arbitrary pattern given as input. The robots are assumed to be autonomous, anonymous and identical. They operate in Look-Compute-Move cycles under an asynchronous scheduler. The robots do not have access to any global coordinate system. The movement of the robots is assumed to be \emph{rigid}, which means that each robot is able to reach its desired destination without interruption. The existing literature that investigates this problem, considers robots with unobstructed visibility. This work considers the problem in the more realistic \emph{obstructed visibility} model, where the view of a robot can be obstructed by the presence of other robots. The robots are assumed to be punctiform and equipped with visible lights that can assume a constant number of predefined colors. We have studied the problem in two settings based on the level of consistency among the local coordinate systems of the robots: \emph{two axis agreement} (they agree on the direction and orientation of both coordinate axes) and \emph{one axis agreement} (they agree on the direction and orientation of only one coordinate axis). In both settings, we have provided a full characterization of initial configurations from where any arbitrary pattern can be formed. 

%

\end{abstract}

\section{Introduction}

One of the recent trends of research in robotics is to use a swarm of simple and inexpensive robots to collaboratively execute complex tasks, as opposed to using a single robot system consisting of one powerful and expensive robot. Robot swarms offer several advantages over single robot systems, such as scalability, robustness and versatility. Algorithmic aspects of decentralized coordination of robot swarms have been extensively studied in the  literature over the last two decades. In theoretical studies, the traditional framework  models the robot swarm as a set of autonomous, anonymous and identical computational entities freely
 moving in the plane. The robots do not have access to any global coordinate system. Each robot is equipped with sensor capabilities to perceive the positions of other robots. The robots operate in \textsc{Look-Compute-Move} (LCM) cycles. In each cycle, a robot takes a snapshot of the positions of the other robots \textsc{(Look)}; based on this snapshot, it executes a deterministic algorithm to determine a destination \textsc{(Compute)}; and it moves towards the computed destination \textsc{(Move)}. Based on the activation and timing of the robots, there are three types of schedulers considered in the literature. In the fully synchronous setting \textsc{(FSync)}, the robots operate in synchronous rounds. All the robots are activated in each round, where they take their snapshots at the same time, and then execute their moves simultaneously. The semi-synchronous model \textsc{(SSync)} is same as the \textsc{FSync} model, except for the fact that not all robots are necessarily activated in each round. The most general model is the asynchronous model \textsc{(ASync)} where there are no assumptions regarding the synchronization and duration of the actions of the robots.

 \textsc{Arbitrary Pattern Formation} or $\mathcal{APF}$ is a fundamental coordination problem in swarm robotics where the robots are required to form any specific but arbitrary geometric pattern given as input. The majority of the literature that studies this problem, considers robots with unlimited and unobstructed visibility. However, if the robots are equipped with camera sensors, then it is unrealistic to assume that each robot can observe all other robots in the swarm, as the line of sight of a robot can be obstructed by the presence of other robots. This setting is known as the \emph{opaque robot} or \emph{obstructed visibility} model. Formally, the robots are modeled as points on the plane and it is assumed that a robot is able to see another robot if and only if no other robot lies on the line segment joining them. Two recent works \cite{FelettiMP18, VaidyanathanST18} investigated two formation problems in this model. In \cite{FelettiMP18}, \textsc{Uniform Circle Formation} problem was studied, where the robots are required to form a circle by positioning themselves on the vertices of a regular polygon. Their approach is to first solve the \textsc{Mutual Visibility} problem as a subroutine where the robots arrange themselves in a configuration in which each robot can see all other robots. Then they solved the original problem from a mutually visible configuration. The more general \textsc{Arbitrary Pattern Formation} problem was first studied in the obstructed visibility model in \cite{VaidyanathanST18}. Their algorithm first solves \textsc{Mutual Visibility} and then solves \textsc{Leader Election} by a randomized algorithm. \textsc{Leader Election} is closely related to $\mathcal{APF}$. For a set of anonymous and identical robots, \textsc{Leader Election} is solved based on geometric characteristics of the configuration. In particular, if the configuration is asymmetric then the asymmetry can be used to elect a unique leader. The same can also be done in certain symmetric configurations. This is easy to do in the unlimited and unobstructed visibility model: the robots take snapshot of the entire configuration and  inspect it to determine the leader. This is not possible in the obstructed visibility model where the robots may not have the full view of the configuration. Therefore, the naive approach would be to first solve \textsc{Mutual Visibility} so that the robots can see the entire configuration. However, since \textsc{Leader Election} and also $\mathcal{APF}$ are deterministically unsolvable from some symmetric configurations, trying to first solve \textsc{Mutual Visibility} may create new symmetries, from where the problem becomes deterministically unsolvable.  Hence, if we want to first bring the robots to a mutually visible configuration, then we have to design a \textsc{Mutual Visibility} algorithm that does not create such symmetries. The existing algorithms in the literature for \textsc{Mutual Visibility} do not have this feature. For this reason, after solving \textsc{Mutual Visibility}, \textsc{Leader Election} was solved in  \cite{VaidyanathanST18} by a randomized algorithm.

 \subsection{Our Contribution}
 
  In this work, our aim is to provide a deterministic solution for $\mathcal{APF}$ in the obstructed visibility model. We consider a system of autonomous, anonymous and identical robots that operate under an asynchronous (\textsc{ASync}) scheduler. We consider the problem in two settings based on the level of consistency among the local coordinate systems of the robots: \emph{two axis agreement}, where the robots agree on the direction and orientation of both coordinate axes and \emph{one axis agreement}, where the robots agree on the direction and orientation of only one coordinate axis, say the $X$-axis. In other words, in the first setting, the robots have a common notion of `up', `down', `left' and `right', while in the second setting, the robots agree on only `left' and `right'. We have shown that if the robots are \emph{oblivious}, i.e., the robots have no memory of their past LCM cycles, then $\mathcal{APF}$ is unsolvable even in the two axis agreement setting. Therefore, we consider the \emph{luminous robots} model in which the robots are equipped with persistent visible lights that can assume a constant number of predefined colors. In this model, we have shown that $\mathcal{APF}$ can be solved from any initial configuration in the two axis agreement setting. However, in the one axis agreement setting, $\mathcal{APF}$ is not deterministically solvable from some initial configurations even for luminous robots with unlimited and unobstructed visibility. In particular, if the robots agree on the direction and orientation of only one coordinate axis, say the $X$ axis, then $\mathcal{APF}$ is unsolvable when the initial configuration has a reflectional symmetry with respect to a line $\mathcal{K}$ which is parallel to the $X$ axis and has no robots lying on $\mathcal{K}$. For all other initial configurations, we show that $\mathcal{APF}$ is solvable.

  Our algorithms require 3 and 6 colors respectively for two axis and one axis agreement. We assume that the movements of the robots are \emph{rigid} in the sense that each robot is able to reach its desired destination without any interruption. Our algorithms work in two stages. In the first stage we solve \textsc{Leader Election}. The leader helps to fix a common global coordinate system. Then in the second stage, the robots arrange themselves to form the input pattern in the global coordinate system. Therefore, this paper also provides a solution for \textsc{Leader Election} in the obstructed visibility model. For robots with obstructed visibility and having only partial agreement in coordinate system, deterministic \textsc{Leader Election} is difficult and is of independent interest. In the two axis agreement setting, \textsc{Leader Election} is trivial: if there is a unique leftmost robot then it becomes the leader and if there are multiple leftmost robots then the bottommost among them becomes the leader. In the one axis agreement setting, the algorithm is much more involved. To the best of our knowledge, this is the first algorithm that solves \textsc{Leader Election} in this model.

 
\subsection{Earlier Works}

The study of \textsc{Arbitrary Pattern Formation}  was initiated in \cite{Yamashita96,Yamashita99}. In these papers, they fully characterized the class of patterns formable by autonomous and anonymous robots with an unbounded amount of memory in \textsc{FSync} and \textsc{SSync}. They characterized the class of formable patterns by using the notion of \emph{symmetricity} which is essentially the order of the cyclic group that acts on the initial configuration. They also showed that in \textsc{SSync}, \textsc{Gathering}, i.e., the point formation problem, for two oblivious robots is unsolvable. However, it is trivially solvable for non-oblivious robots. This differentiates non-oblivious \textsc{SSync} robots from oblivious ones. Later, in \cite{Yamashita10}, they characterized the families of patterns formable by \emph{oblivious} robots in \textsc{FSync} and \textsc{SSync}. The problem was first studied in the weak setting of oblivious robots in \textsc{ASync} in \cite{Flocchini08}. They showed that if the robots have no common agreement on coordinate system, then it is impossible to form an arbitrary pattern. If the robots have one axis agreement, then any odd number of robots can form an arbitrary pattern, but an even number of robots cannot, in the worst case. If the robots agree on both $X$ and $Y$ axes, then any pattern is formable from any configuration of robots. They also proved that it is possible to elect a leader for $n \geq 3$ robots if it is possible to form any arbitrary pattern. In \cite{Petit09,Dieudonne10}, the authors studied the relationship between \textsc{Arbitrary Pattern Formation} and \textsc{Leader election}.  They proved that any arbitrary pattern can be formed from any initial configuration wherein the leader election is possible. More precisely, their algorithms work for four or more robots with chirality and for at least five robots without chirality. Having chirality means that the robots agree on a cyclic orientation. Combined with the result in \cite{Flocchini08}, it follows that \textsc{Arbitrary Pattern Formation} and \textsc{Leader election} are equivalent, i.e., it is possible to solve \textsc{Arbitrary Pattern Formation} for $n \geq 4$ with chirality (resp. $n \geq 5$ without chirality) if and only if \textsc{Leader election} is solvable.  Recently, the case of $n=4$ robots was fully characterized in \cite{Bramas18} in \textsc{ASync}, with and without agreement in chirality. They proposed a new geometric invariant that exists in any configuration with four robots, and using this invariant, they presented an algorithm that forms any target pattern from any solvable initial configuration. In \cite{Cicerone17,Fujinaga15}, the problem was studied in \textsc{ASync} allowing the pattern to have multiplicities. The problem of forming a sequence of patterns in a given order was studied in \cite{Das15}. Randomized algorithms for pattern formation were studied in \cite{Yamauchi14, BramasT16}. In \cite{cicerone2018embedded,Fujinaga15}, the so-called \textsc{Embedded Pattern Formation} problem was studied where the pattern to be formed is provided as a set of fixed and visible points in the plane. In \cite{bose18}, the problem was considered in a grid based terrain where the movements of the robots are restricted only along grid lines and only by a unit distance in each step. They showed that a set of asynchronous robots without agreement in coordinate system can form any arbitrary pattern, if their starting configuration is asymmetric. 


All the aforementioned works considered robots with unlimited and unobstructed visibility. In \cite{Yamauchi13}, the problem was first studied in the the limited, but unobstructed, visibility setting, where a robot can see all robots that are within a fixed limited distance $V$. They first showed that oblivious \textsc{FSync} robots with limited visibility cannot solve $\mathcal{APF}$. Therefore, they considered non-oblivious robots with unlimited memory, each of which can record the history of local views and outputs during execution. For these robots, they presented algorithms that work in \textsc{FSync} with non-rigid movements and  \textsc{SSync} with rigid movements. Their strategy is to first converge the robots towards a point (using the algorithm from \cite{AndoOSY99}) until they all observe each other. Once this is achieved, the situation reduces to the full visibility setting and the problem can be solved using existing techniques. However, the convergence phase can create new symmetries in the configuration. These symmetries are broken using the local memory of the robots in which all the snapshots taken since the first round are stored. Pattern formation in the limited visibility setting has been considered recently in \cite{Lukovszki14} for position aware robots. The paper considers \textsc{FSync} robots with constant size memory operating on an infinite grid with a common coordinate system. Furthermore, robots are given a fixed point on the grid so that they can form a connected configuration containing it.  In the obstructed visibility model, the \textsc{Gathering} problem has been studied for robots in three dimensional space \cite{Bhagat18}. \textsc{Gathering} has also been studied in the \emph{opaque fat robots} model in the plane \cite{agathangelou2013distributed}. In this model, the robots are modeled as disks of the same size and it is assumed that any two robots $r$ and $r'$ can see each other if and only if there exist point $p$ and $p'$ on the circumference of $r$ and $r'$ respectively such that the line segment joining $p$ and $p'$ does not intersect any other robot point. A related problem in obstructed visibility model is the \textsc{Mutual Visibility} problem. In this problem, starting from arbitrary configuration, the robots have to reposition themselves to a configuration in which every robot can see all other robots in the team.  The problem has been extensively studied in the literature under various settings \cite{di2017mutual, sharma2017log,sharma2016complete,bhagat2017optimum,sharma2015mutual,Adhikary18}. \textsc{Arbitrary Pattern Formation} in the obstructed visibility model was first studied recently in \cite{VaidyanathanST18}, without any agreement in coordinate system, where the authors proved runtime bounds in terms of the time required to solve \textsc{Leader election}. However, they did not provide any deterministic solution for \textsc{Leader election} and it is yet to be studied in the literature in the obstructed visibility model.

\subsection{Outline}

 The paper is organized as follows. In Section \ref{sec:model}, the robot model under study is presented. In Section \ref{sec: preli}, the formal definition of the problem is given along with some basic notations and terminology. In Section \ref{sec: imp}, we present some impossibility results. In Section \ref{sec: one axis}, we describe the algorithm for $\mathcal{APF}$ under one axis agreement. In Section \ref{sec: main result and proof}, the main results of the paper are presented along with formal proofs.

\section{Robot Model}\label{sec:model}

A set of $n$ mobile computational entities, called \emph{robots}, are initially placed at distinct points in the Euclidean plane. Each robot can move freely in the plane. The robots are assumed to be:
\begin{description}
  \item [anonymous:] they have no unique identifiers that they can use in a computation
 
 \item [identical:] they are indistinguishable by their physical appearance
 
 \item [autonomous:] there is no centralized control
 
 \item [homogeneous:] they execute the same deterministic algorithm.
\end{description}

The robots are modeled as points in the plane, i.e., they do not have any physical extent. The robots do not have access to any global coordinate system. Each robot is provided with its own local coordinate system centered at its current position, and its own notion of unit distance and handedness. However, the robots may have a priori agreement about the direction and orientation of
the axes in their local coordinate system. Based on this, we consider the following two models.
\begin{description}
  \item [Two axis agreement:] They agree on the direction and orientation of both axes.
 
 \item [One axis agreement:] They agree on the direction and orientation of only one axis. 
 
\end{description}

The \emph{opaque robot} or \emph{obstructed visibility} model assumes that visibility of a robot can be obstructed by the presence of other robots. We assume that the robots have unlimited but obstructed visibility, i.e., the robots are equipped with a 360 vision camera sensor that enables the robots to take snapshots of the entire plane, but its vision can be obstructed by the presence of other robots. Formally, a point $p$ on the plane is visible to a robot $r$ if and only if the line segment joining $p$ and $r$ does not contain any other robot. Hence, two robots $r_1$ and $r_2$ are able to see each other if and only if no other robot lies on the line segment joining them.

This paper studies the pattern formation problem in the \emph{robots with lights} or \emph{luminous robot} model, introduced by Peleg \cite{peleg2005distributed}. In this model, each robot is equipped with a visible light which can assume a constant number of predefined colors. The lights serve both as a weak explicit communication mechanism and a form of internal memory. We denote  the set of colors available to the robots by $\mathcal{C}$. Notice that a robot having light with only one possible color has the same capability as the one with no light. Therefore, the luminous robot model generalizes the classical model.

The robots, when active, operate according to the so-called \emph{LOOK-COMPUTE-MOVE} \emph{(LCM)} cycle. In each cycle, a previously idle or inactive robot wakes up and executes the following steps.
\begin{description}
 \item[LOOK:] The robot takes a snapshot of the current configuration, i.e., it obtains the positions, expressed in its local
coordinate system, of all robots visible to it, along with their respective colors. The robot also knows its own color.
 
 \item[COMPUTE:] Based on the perceived configuration, the robot performs computations according to a deterministic algorithm to decide a destination point $x \in \mathbb{R}^2$ (expressed in its local coordinate system) and a color $c \in \mathcal{C}$. As mentioned earlier, the deterministic algorithm is  same for all robots. The robot then sets its light to $c$.
 
 \item[MOVE:] The robot then moves towards the point $x$.
\end{description}

After executing a LOOK-COMPUTE-MOVE cycle, a robot becomes inactive. Then after some finite time, it wakes up again to perform another LOOK-COMPUTE-MOVE cycle. Notice that after a robot sets its light to a particular color in the COMPUTE phase of a cycle, it maintains its color until the COMPUTE phase of the next LCM cycle. When a robot transitions from one LCM cycle to the next, all of its local memory (past computations and snapshots) are erased, except for the color of the light.

Based on the activation and timing of the robots, there are three types of schedulers considered in the literature.
\begin{description}
 \item[Fully synchronous:] In the fully synchronous setting \textsc{(FSync)}, time can be logically divided into global rounds. In each round, all the robots are activated. They take the snapshots at the same time, and then perform their moves simultaneously. 
 
 \item [Semi-synchronous:] The semi-synchronous setting \textsc{(SSync)} coincides with the \textsc{FSync} model, with the only difference that not all robots are necessarily activated in each round. However, every robot is activated infinitely often.
 
 \item[Asynchronous:] The asynchronous setting \textsc{(ASync)} is the most general model. The robots are activated independently and each robot executes its cycles independently. The amount of time spent in LOOK, COMPUTE, MOVE and inactive states is finite but unbounded, unpredictable and not same for different robots. As a result, the robots do not have a common notion of time. Moreover, a robot can be seen while moving, and hence, computations can be made based on obsolete information about positions. Also, the configuration perceived by a robot during the LOOK phase may significantly change before it makes a move and therefore, may cause a collision. 

\end{description}

 The scheduler that controls the activations and the durations of the operations can be thought of as an adversary. In the \textsc{Non-Rigid} movement model, the adversary also has the power to stop the movement of a robot before it reaches its destination. However, $\exists$ $\delta > 0$ so that each robot traverses at least the distance $\delta$ unless its destination is closer than $\delta$. This restriction imposed on the adversary is necessary, because otherwise, the adversary can stop a robot after moving distances $\frac{1}{2}, \frac{1}{4}, \frac{1}{8}, \ldots$ and thus, not allowing any robot to traverse a distance of more than 1. In the \textsc{Rigid} movement model, each robot is able to reach its desired destination without any interruption. In this paper, the robots are assumed to have \textsc{Rigid} movements. The adversary also has the power to choose the local coordinate systems of individual robots (obeying the agreement assumptions). However, for each robot, the local coordinate system set by the adversary when it is first activated, remains unchanged. In other words, the local coordinate systems of the robots are persistent. Of course, in any COMPUTE phase, if instructed by the algorithm, a robot can logically define a different coordinate system, based on the snapshot taken, and transform the coordinates of the positions of the observed robots in the new coordinate system. This coordinate system is obviously not retained by the robot in the next LCM cycle.

\section{Preliminaries} \label{sec: preli}

\subsection{Definitions and Notations}

We assume that a set $\mathcal{R}$ of $n$ robots are placed at distinct positions in the Euclidean plane. For any time $t$, $\mathbb{C}(t)$ will denote the configuration of the robots at time $t$. For a robot $r$, its position at time $t$ will be denoted by $r(t)$. When there is no ambiguity, $r$ will represent both the robot and the point in the plane occupied by it. By $1_r$, we shall denote the unit distance according to the local coordinate system of $r$. At any time $t$, $r(t).light$ or simply $r.light$ will denote the color of the light of $r$ at $t$.

Suppose that a robot $r$, positioned at point $p$, takes a snapshot at time $t_1$. Based on this snapshot, suppose that the deterministic algorithm run in the COMPUTE phase instructs the robot to change its color (Case 1) or move to a different point (Case 2) or both (Case 3). In case 1, assume that it changes its color at time $t_2 > t_1$. In case 2, assume that it starts moving at time $t_3 > t_1$. Note that when we say that it starts moving at $t_3$, we mean that $r(t_3) = p$, but $r(t_3 + \epsilon) \neq p$ for sufficiently small $\epsilon > 0$. 
For case 3, assume that $r$ changes it color at  $t_2 > t_1$ and starts moving at  $t_3 > t_2$. Then we say that $r$ has a \emph{pending move} at $t$ if $t \in (t_1, t_2)$ in case 1 or $t \in (t_1, t_3]$ in case 2 and 3. A robot $r$ is said to be \emph{stable} at time $t$, if $r$ is stationary and has no pending move at $t$. A configuration at time $t$ is said to be a \emph{stable configuration} if every robot is stable at $t$. A configuration at time $t$ is said to be a \emph{final configuration} if 
\begin{enumerate}                                                                                                                                                                                                                                                                                                \item every robot at $t$ is stable,                                                                                                                                                                                                                                                                                                                                                                                                                                                                                                                                                                                   \item any robot taking a snapshot at $t$ will not decide to move or change its color.                                                                                                                                                                                                                                                                                                 \end{enumerate}

With respect to the local coordinate system of a robot, positive and negative directions of the $X$-axis will be referred to as \emph{right} and \emph{left} respectively, while the positive and negative directions of the $Y$-axis will be referred to as \emph{up} and \emph{down} respectively. Furthermore, any straight line parallel to the $X$-axes ($Y$-axes) will be referred to as a horizontal (resp. vertical) line. For a robot $r$, $\mathcal{L}_{V}(r)$ and $\mathcal{L}_{H}(r)$ are respectively the vertical and horizontal line passing through $r$. We denote by $\mathcal{H}_{U}^{O}(r)$ (resp. $\mathcal{H}_{U}^{C}(r)$) and $\mathcal{H}_{B}^{O}(r)$ (resp. $\mathcal{H}_{B}^{C}(r)$) the upper and bottom open (resp. closed) half-plane delimited by $\mathcal{L}_{H}(r)$ respectively. Similarly, $\mathcal{H}_{L}^{O}(r)$ (resp. $\mathcal{H}_{L}^{C}(r)$) and $\mathcal{H}_{R}^{O}(r)$ (resp. $\mathcal{H}_{R}^{C}(r)$) are the left and right open (resp. closed) half-plane delimited by $\mathcal{L}_{V}(r)$ respectively. We define $\mathcal{H}^{O}(r_1, r_2)$ (resp. $\mathcal{H}^{C}(r_1, r_2)$) as the horizontal open (resp. closed) strip delimited by $\mathcal{L}_{H}(r_1)$ and $\mathcal{L}_{H}(r_2)$. For a robot $r$ and a straight line $\mathcal{L}$ passing through it, $r$ will be called \emph{non-terminal on $\mathcal{L}$} if it lies between two other robots on $\mathcal{L}$, and otherwise it will be called \emph{terminal on $\mathcal{L}$}.

\subsection{Problem Definition}\label{sec:preli}

A swarm of $n$ robots is arbitrarily deployed at distinct positions in the Euclidean plane. Initially, the lights of all the robots are set to a specific color called \texttt{off}. Each robot is given a pattern $\mathbb{P}$ as input, which is a list of $n$ elements from $\mathbb{R}^2$. Our algorithms will also work if the given pattern $\mathbb{P}$ has multiplicities. However, for simplicity we assume that $\mathbb{P}$ has no multiplicities, i.e., the $n$ elements in the list are all distinct.  Notice that the robots, through the input, implicitly obtain $n$, the total number of robots. Without loss of generality, we assume that 

\begin{enumerate}
 \item the $n$ elements in $\mathbb{P}$ are arranged in lexicographic order, i.e., $(x,y) < (x',y')$ iff either $x < x'$ or $x = x'$, $y < y'$,
 
 \item $(x,y) \in \mathbb{R}^2_{\geq 0} = \{ (a,b) \in \mathbb{R}^2 | a, b \geq 0 \}$ for each element $(x,y)$ in $\mathbb{P}$.
\end{enumerate}

The goal of the \textsc{Arbitrary Pattern Formation}, or in short $\mathcal{APF}$, is to design a distributed algorithm so that there is a  time $t$ such that

\begin{enumerate}
 \item $\mathbb{C}(t)$ is a final configuration,
 
 \item  the lights of all the robots at $t$ are set to the same color,
 
 \item $\mathbb{P}$ can be obtained from $\mathbb{C}(t)$ by a sequence of translation, reflection, rotation, uniform scaling operations.
 
\end{enumerate}

Furthermore, it is required that the movements must be \emph{collisionless} in the sense that at any time $t' \in [0,t]$, no two robots  occupy the same position in the plane.

\section{Some Impossibility Results}\label{sec: imp}

In this section, we prove some impossibility results. We first show that in the obstructed visibility the model, $\mathcal{APF}$ is unsolvable by oblivious robots. The result holds even for robots with two axis agreement.

\begin{theorem}\label{thm: theory1}
 $\mathcal{APF}$ is unsolvable by oblivious opaque robots even with two axis agreement.
\end{theorem}

\begin{proof}
 For the sake of contradiction, assume that there is a distributed algorithm $\mathcal{A}$ that solves $\mathcal{APF}$ for oblivious opaque robots. Consider a swarm of $n = 3k+2$, $k > 1$ robots. Suppose that the given pattern $\mathbb{P}$ is the following: $\mathbb{P}[0] = (0,0), \mathbb{P}[n-1] = (4k+2, 0)$ and for $j = 1, \ldots, k,$ $\mathbb{P}[3j-2] = (4j-2,0)$, $\mathbb{P}[3j-1] = (4j-1,0)$, $\mathbb{P}[3j] = (4j,0)$. Here $\mathbb{P}[i]$ denotes the $i$th element of $\mathbb{P}$. See Fig. \ref{fig: 1theory1} for an instance with $k = 4$, i.e., $n = 14$. Now assume that starting from some initial configuration, the algorithm $\mathcal{A}$ forms a final configuration $\mathbb{C}$ which is similar to $\mathbb{P}$, i.e., there is a coordinate system $C$ such that the robots, say $r_0, r_1, \ldots, r_{n-1}$, are at $\mathbb{P}[0], \mathbb{P}[1], \ldots, \mathbb{P}[n-1]$ respectively. Partition the robots as $S_0 = \{r_0\}, S_1 = \{r_{i} | i = 3j-2, j = 1, \ldots, k\}, S_2 = \{r_{i} | i = 3j-1, j = 1, \ldots, k\}, S_3 = \{r_{i} | i = 3j, j = 1, \ldots, k\}, S_4 = \{r_{n-1}\}$. It is easy to see that the local views of the robots belonging to the same partition are identical. Also, since the configuration is final and an algorithm for oblivious robots depends only on the local view of a robot and the input pattern, for each of the five types of views, the algorithm $\mathcal{A}$ instructs to do nothing. Now consider an initial configuration $\mathbb{C}'$ in which the robots are at points $p_0, p_1, \ldots, p_{n-1}$ such that with respect to the coordinate system $C$, $p_0 = (0,0), p_{n-1} = (3k+3,0)$ and for $j = 1, \ldots, 3k$, $p_j = (1+j,0)$. See Fig. \ref{fig: 1theory2}. Notice that the views of the robots at $p_j$, $j = 2, \ldots, 3k-1$, in $\mathbb{C}'$ are identical to the views of the robots in $S_2$ in $\mathbb{C}$. Also the views of the robots at $p_0$, $p_1$, $p_{n-2}$ and $p_{n-1}$ in $\mathbb{C}'$ are respectively identical to the views of the robots in $S_0, S_1, S_3$ and $S_{4}$ in $\mathbb{C}$. Therefore, the robots in $\mathbb{C}'$ will remain stationary according to algorithm $\mathcal{A}$ and hence, the given pattern $\mathbb{P}$ can not be formed.      
\end{proof}

 \begin{figure}[h]
\centering
\subcaptionbox[Short Subcaption]{ The final configuration $\mathbb{C}$ similar to the input pattern $\mathbb{P}$
       \label{fig: 1theory1}
}
[
    0.6\textwidth 
]
{
    \fontsize{8pt}{8pt}\selectfont
    \def\svgwidth{0.6\textwidth}
    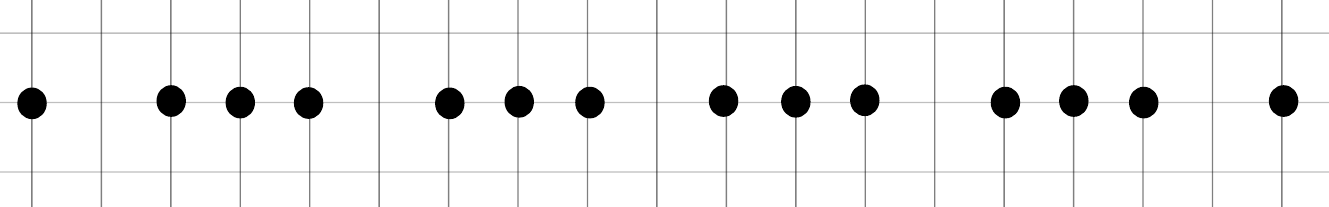
}
\\
\subcaptionbox[Short Subcaption]{ The initial configuration $\mathbb{C}'$.
       \label{fig: 1theory2}
}
[
    0.6\textwidth 
]
{
    \fontsize{8pt}{8pt}\selectfont
    \def\svgwidth{0.6\textwidth}
    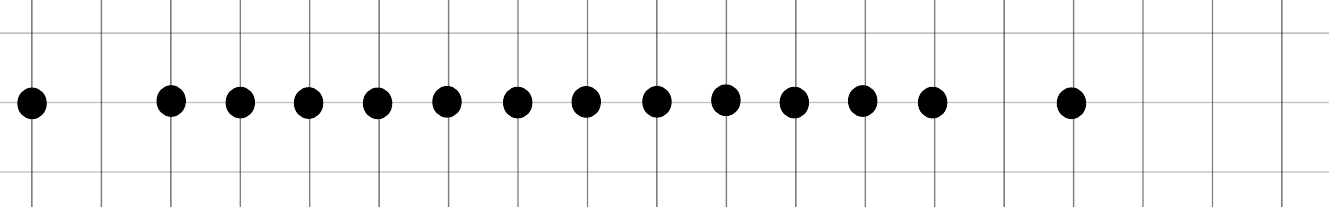
}

\caption[Short Caption]{Illustrations supporting the proof of Theorem \ref{thm: theory1}.}
\label{}
\end{figure}

Next we show that even for robots with lights and full visibility, $\mathcal{APF}$ is unsolvable in one axis agreement setting if the initial configuration has certain symmetries. The result is proved using the arguments in the proof of Theorem 2.2 in \cite{Flocchini08}.

\begin{theorem}\label{thm: theory}
 Consider a set of luminous robots that agree on the direction and orientation of only one axis, say $X$ axis. Then $\mathcal{APF}$ is unsolvable even in \textsc{FSync} with full visibility if the initial configuration has a reflectional symmetry with respect to a line $\mathcal{K}$ which is parallel to the $X$ axis with no robots lying on $\mathcal{K}$.
\end{theorem}

\begin{proof}
 Assume that the initial configuration has a reflectional symmetry with respect to a horizontal line $\mathcal{K}$ and there are no robots lying on $\mathcal{K}$. Initially, the lights of all the robots are set to a specific color called \texttt{off}. Let $\mathcal{H}$ and $\mathcal{H}'$ be the two open half-planes delimited by $\mathcal{K}$. Assume that the adversary sets the $Y$-axes of the local coordinate systems of the robots of $\mathcal{H}$ and $\mathcal{H}'$ in opposite directions (see Fig. \ref{fig: theory}). We show that in each round, 1) the configuration is symmetric with respect to $\mathcal{K}$, 2) there are no robots lying on $\mathcal{K}$, and 3) if $r$ and $r'$ are in symmetric positions with respect to $\mathcal{K}$, then their lights are set to the same color and the $Y$-axes of their local coordinate systems are in opposite directions. It follows from our assumptions about the initial configuration that 1)-3) are true in round 1. Now assume that 1)-3) are true in round $t$, $t \geq 1$. The robots in round $t$ can be partitioned into $\frac{n}{2}$ disjoint pairs $\{r, r'\}$, where $r$ is in $\mathcal{H}$ and $r'$ is its specular partner in  $\mathcal{H}'$. Consider a pair of specular robots $r$ and $r'$. Assume that $r$ decides to move to a point $p$. 
 
 \textbf{Case 1:} Let $p \in \mathcal{K}$. Then $r'$ also decides to move to $p$, as they have the same view, same color, and they both execute the same deterministic algorithm. Hence, they will collide at $p$.
 
 \textbf{Case 2:} Let $p \in \mathcal{H'}$. By the same logic, $r'$ will decide to move to $p' \in \mathcal{H}$, which is the mirror image of $p$ with respect to $\mathcal{K}$. Similar to the previous case, they will collide on $\mathcal{K}$ (see Fig. \ref{fig: theory 2}).
 
 \textbf{Case 3:} Therefore, the only possible case is $p \in \mathcal{H}$. Again, $r'$ will decide to move to $p' \in \mathcal{H}'$, which is the mirror image of $p$ with respect to $\mathcal{K}$. Also, if they change their lights, both will set the same color. 
 
 Since the same happens for all pairs of specular robots, at the end of the round, the configuration remains symmetric with respect to $\mathcal{K}$, with $\mathcal{K}$ containing no robots and the specular partners having same colors and the $Y$-axes of their local coordinate systems in opposite directions (see Fig. \ref{fig: theory 3}). Hence, 1)-3) are true in round $t+1$. 
 
 Therefore, it is easy to see that it is impossible to form an arbitrary pattern, for example, an asymmetric pattern. Hence, $\mathcal{APF}$ is unsolvable.   
\end{proof}

 \begin{figure}[h]
\centering
\subcaptionbox[Short Subcaption]{
       \label{fig: theory}
}
[
    0.32\textwidth 
]
{
    \fontsize{8pt}{8pt}\selectfont
    \def\svgwidth{0.32\textwidth}
    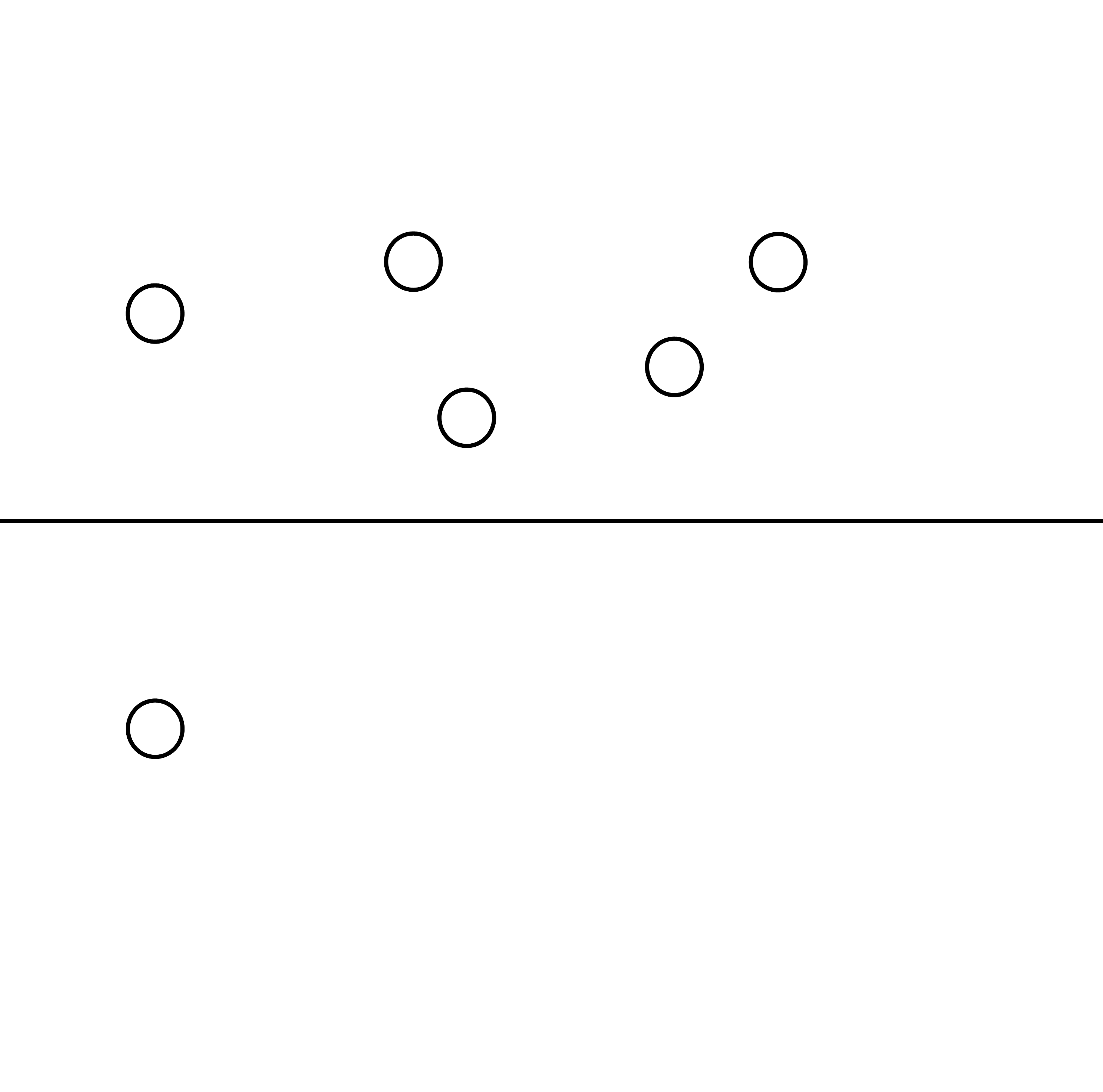
}
\hfill
\subcaptionbox[Short Subcaption]{
     \label{fig: theory 2}
}
[
    0.32\textwidth 
]
{
    \fontsize{8pt}{8pt}\selectfont
    \def\svgwidth{0.32\textwidth}
    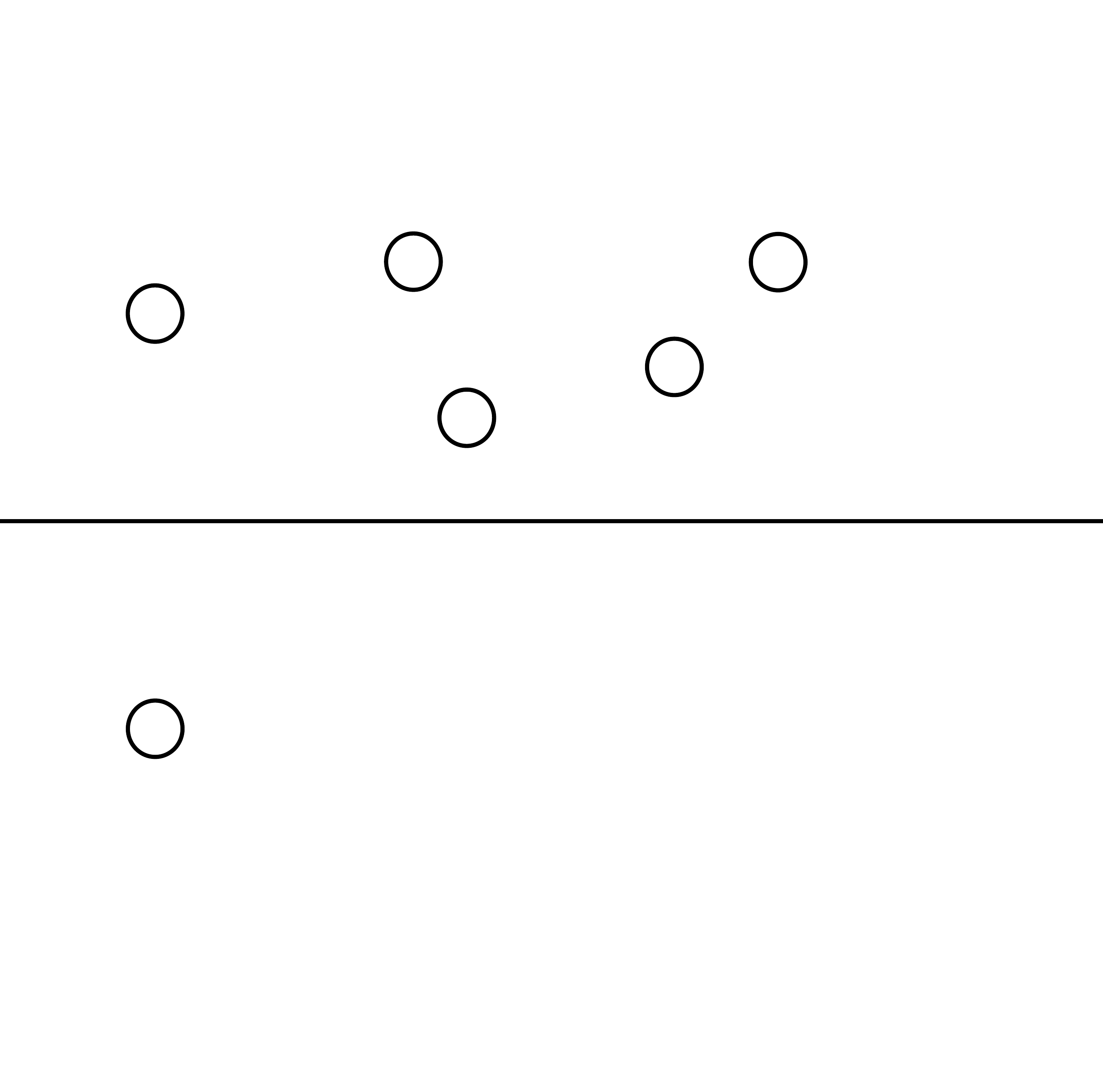
}
\hfill
\subcaptionbox[Short Subcaption]{
       \label{fig: theory 3}
}
[
    0.32\textwidth 
]
{
    \fontsize{8pt}{8pt}\selectfont
    \def\svgwidth{0.32\textwidth}
    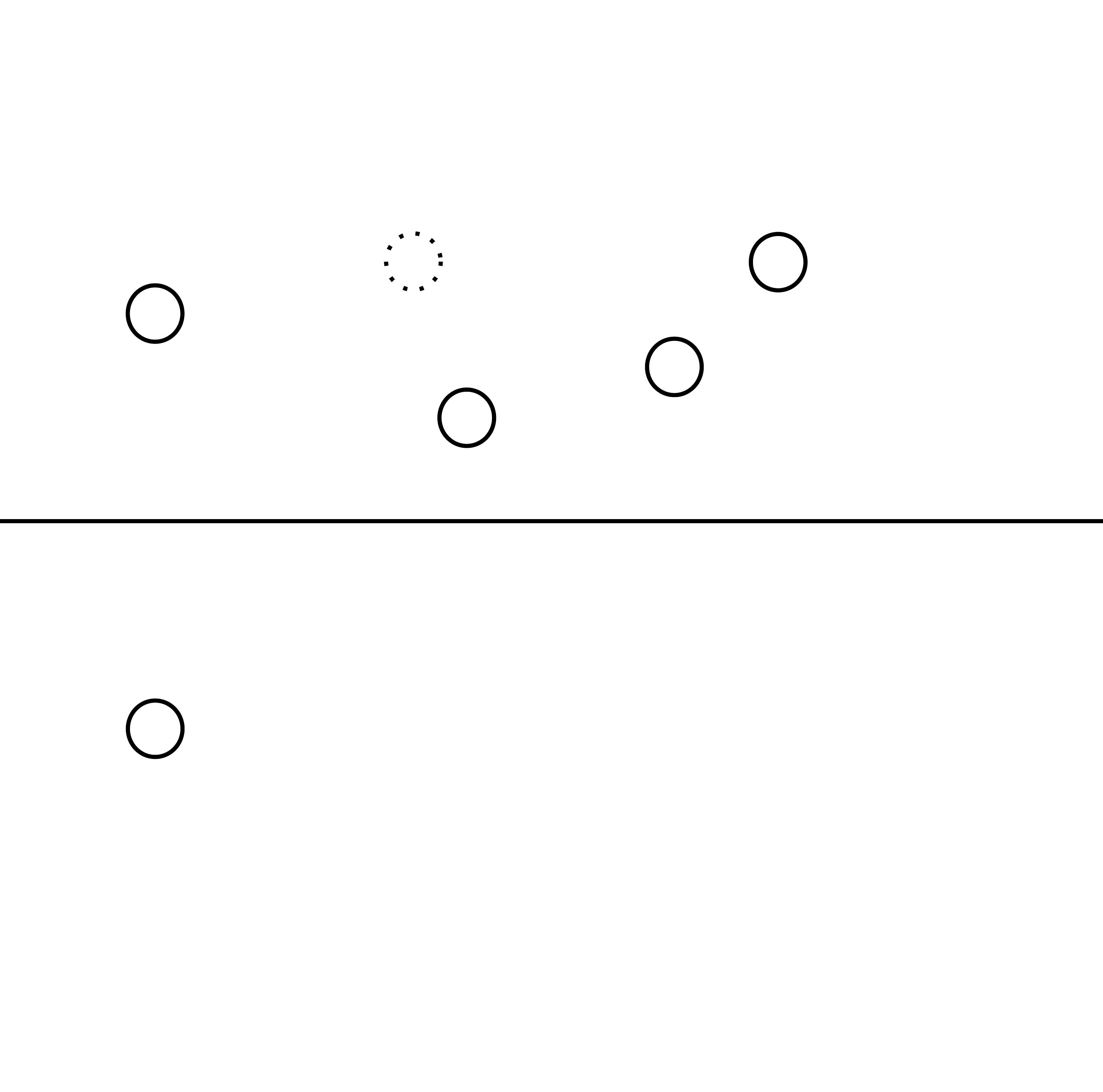
}

\caption[Short Caption]{Illustrations supporting the proof of Theorem \ref{thm: theory}.}
\label{fig: compute destination}
\end{figure}

\section{Arbitrary Pattern Formation under One Axis Agreement}\label{sec: one axis}

In this section, we shall discuss $\mathcal{APF}$ under the one axis agreement model in \textsc{ASync}. We assume that the robots agree on the direction and orientation of only $X$ axis. According to Theorem \ref{thm: theory}, $\mathcal{APF}$ can not be solved by a deterministic algorithm if the initial configuration is symmetric with respect to a line that is parallel to $X$ axis and does not contain any robot on it. Therefore, we shall assume that the initial configuration does not have such a symmetry. We shall prove that with this assumption, $\mathcal{APF}$ is solvable in \textsc{ASync}. Our algorithm requires six colors, namely \texttt{off}, \texttt{terminal}, \texttt{candidate}, \texttt{symmetry}, \texttt{leader}, and \texttt{done}. As mentioned earlier, initially the lights of all the robots are set to \texttt{off}.

The goal of $\mathcal{APF}$ is that the robots have to arrange themselves to a configuration which is similar to $\mathbb{P}$ with respect to translation, rotation, reflection and uniform scaling. Since the robots do not have access to any global coordinate system, there is no agreement regarding where and how the pattern $\mathbb{P}$ is to be embedded in the plane. To resolve this ambiguity, we shall first elect a robot in the team as the leader. The relationship between leader election and arbitrary pattern formation is well established in the literature. Once a leader is elected, it is not difficult to reach an agreement on a suitable coordinate system. Then the robots have to reconfigure themselves to form the given pattern $\mathbb{P}$ with respect to this coordinate system. Thus, the algorithm is divided into two \emph{stages}, namely \emph{leader election} and \emph{pattern formation from leader configuration}. The leader election stage is again logically divided into two \emph{phases}, \emph{phase 1} and \emph{phase 2}. In each LCM cycle, it has to infer in which stage or phase it currently is, from certain geometric conditions and the lights of the robots in the perceived configuration. These conditions are described in Algorithm \ref{main_algorithm}. Notice that due to the obstructed visibility, two robots taking snapshots at the same time can have quite different views of the configuration. Therefore, it may happen that they decide to execute instructions corresponding to different stages or phases of the algorithm. 
$\\  \\$
  \begin{algorithm}[H]
    \setstretch{.1}
    \SetKwInOut{Input}{Input}
    \SetKwInOut{Output}{Output}
    \SetKwProg{Fn}{Function}{}{}
    \SetKwProg{Pr}{Procedure}{}{}

    \Input{The configuration of robots visible to me.}
    
    \Pr{\textsc{ArbitraryPatternFormation()}}{

    \uIf(\tcp*[f]{stage 2}){there is a robot with light set to \texttt{leader}}{
    
	\textsc{PatternFormationFromLeaderConfiguration()}
    
      }
      
      \Else(\tcp*[f]{stage 1}){
    
	\uIf{there are two robots with light set to \texttt{candidate} on the same vertical line or at least one robot with light set to \texttt{symmetry}}{
	
	    \textsc{Phase2()} }
	    
	\Else{\textsc{Phase1()}}

    }
    
    }
\caption{Arbitrary Pattern Formation}
    \label{main_algorithm} 
\end{algorithm}

\subsection{Leader Election}\label{sec:election}

For a group of anonymous and identical robots, leader election is solved on the basis of the relative positions of the robots in the configuration. But this is only possible if the robots can see the entire configuration. Therefore, the naive approach would be to first bring the robots to a mutually visible configuration where each robot can see all other robots. But as mentioned earlier, this can create unwanted symmetries in the configuration from where arbitrary pattern formation may be unsolvable. Therefore, we shall employ a different strategy that does not require solving mutual visibility.

The aim of the leader election stage is to obtain a stable configuration where there is a unique robot $r_l$ with its light set to \texttt{leader} and the remaining $n-1$ robots have their lights set to \texttt{off}. But we would like the configuration to satisfy some additional properties that will be useful in the next stage. Formally, our aim is to obtain a stable configuration where there is a unique robot $r_l$ such that
\begin{enumerate}
   \item $r_l.light =$ \texttt{leader}
   
   \item $r.light = $ \texttt{off} for all $r \in$ $\mathcal{R} \setminus \{r_l\}$
   
   \item $r \in \mathcal{H}_{R}^{O}(r_l) \cap \mathcal{H}_{U}^{O}(r_l)$ for all $r \in$ $\mathcal{R} \setminus \{r_l\}$.
\end{enumerate}
We shall call this a \emph{leader configuration}, and call $r_l$ the \emph{leader}. As mentioned earlier, the leader election algorithm consists of two phases, namely phase 1 and phase 2 (see Algorithm \ref{main_algorithm}). We describe the phases in detail in the following.

\subsubsection{Phase 1}\label{sec: phase 1}

\paragraph{Overview.}

Since the robots agree on left and right, if there is a unique leftmost robot in the configuration, then it can identify this from its local view and elect itself as the leader.  However, there might be more than one leftmost robots in the configuration. Assume that there are $k \geq 2$ leftmost robots in the configuration. Our aim in this phase is to reduce the number of leftmost robots to $k = 2$, or if possible, to $k = 1$. Suppose that the $k \geq 2$ leftmost robots are lying on the vertical line $\mathcal{L}$. We shall ask the two terminal robots on $\mathcal{L}$, say $r_1$ and $r_2$, to move horizontally towards left. If both robots move synchronously by the same distance, then the new configuration will have two leftmost robots. However, $r_1$ and $r_2$ can not distinguish between this configuration and the initial configuration, i.e., they can not ascertain from their local view if they are the only robots on the vertical line due to the obstructed visibility. To resolve this, we shall ask $r_1$ and $r_2$ to change lights to \texttt{terminal} before moving. Now, consider the case where $r_1$ and $r_2$ move different distances with $r_1$ moving further. Suppose that $r_2$ reaches its destination first, and when it takes the next snapshot, it finds that $r_1$ (with light \texttt{terminal}) is on the same vertical line. So, $r_2$ incorrectly concludes that the two terminal robots have been brought on the same vertical line, while actually $r_1$ is still moving leftwards. To avert this situation, we shall use the color \texttt{candidate}. After moving, the robots will change their lights to \texttt{candidate} to indicate that they have completed their moves. So, if we have two robots with light \texttt{candidate} on the same vertical line, then we are done. On the other hand, if we end up with the two robots with light \texttt{candidate} not on the same vertical line, then the one on the left will become the leader.

\paragraph{Formal Description.}

Before describing the algorithm, let us give a definition. We define a \emph{candidate configuration} to be a stable configuration where there are two robots $r_1$ and $r_2$ such that
\begin{enumerate}
   \item $r_1.light = r_2.light =$ \texttt{candidate}
      
   \item $r.light =$ \texttt{off} for all $r \in$ $\mathcal{R} \setminus \{r_1,r_2\}$
   
   \item  $r_1$ and $r_2$ are on the same vertical line
   
   \item $r \in \mathcal{H}_{R}^{O}(r_1)$ for all $r \in$ $\mathcal{R} \setminus \{r_1, r_2\}$.
\end{enumerate}

 Our aim in this phase is to create either a leader configuration or a candidate configuration. A pseudocode description of phase 1 is given in Algorithm \ref{algo:phase1}.
 $\\  \\$
 \begin{algorithm}[H]
    \setstretch{.1}
    \SetKwInOut{Input}{Input}
    \SetKwInOut{Output}{Output}
    \SetKwProg{Fn}{Function}{}{}
    \SetKwProg{Pr}{Procedure}{}{}

    \Input{The configuration of robots visible to me.}
    
    \Pr{\textsc{Phase1()}}{

    $r \leftarrow$ myself
    
    \uIf{$r.light =$ \texttt{off}}{
    
	\uIf{there are no robots in $\mathcal{H}_{L}^{C}(r)$ other than itself and all robots have their lights set to \texttt{off}}{
	
	  \textsc{BecomeLeader()}
	
	}
    
	\ElseIf{\textsc{LeftMostTerminal()} = True}{
	
	    $(x^*,y^*) \leftarrow$ \textsc{ComputeDestination()}
	    
	    $r.light \leftarrow$ \texttt{terminal}
	    
	    Move to $(x^*,y^*)$

	    }

      }

      \uElseIf{$r.light =$ \texttt{terminal}}{
	$r.light \leftarrow $ \texttt{candidate}
      }

      \ElseIf{$r.light = $ \texttt{candidate}}{

      \uIf{there is a robot with light \texttt{candidate} in $\mathcal{H}_{R}^{O}(r)$}{
      
	$r.light \leftarrow$ \texttt{off}
      
      }
      
      \ElseIf{there is a robot with light \texttt{off} in $\mathcal{H}_{L}^{O}(r)$}{
      
	$r.light \leftarrow$ \texttt{off}
      
      }
      
      }

    }
    

\Pr{\textsc{BecomeLeader()}}{

    \uIf{there are no robots in $\mathcal{H}_{B}^{C}(r)$ other than itself}{
	    
	    $r.light \leftarrow$ \texttt{leader}
	    
	  }
	  
	  \Else{
	  
	  $r' \leftarrow$ the bottommost robot
	  
	  $d \leftarrow$ $|r'.y|$
	  
	  Move $d+1_r$ distance down vertically 
	  }

    }


\Fn{\textsc{LeftMostTerminal()}}{

    $result \leftarrow False$
    
    \uIf{there are no robots in $\mathcal{H}_{L}^{O}(r)$}{
    
      \If{there is $\mathcal{H} \in \{\mathcal{H}_{B}^{O}(r), \mathcal{H}_{U}^{O}(r)\}$ such that $\mathcal{H} \cap \mathcal{L}_{V}(r)$ contains no robots}{$result \leftarrow True$}
    
    }
    
    \ElseIf{there is exactly one robot $r'$ in $\mathcal{H}_{L}^{O}(r)$ and $r'.light$ = \texttt{candidate}}{
    
      Let $\mathcal{H} \in \{\mathcal{H}_{B}^{O}(r), \mathcal{H}_{U}^{O}(r)\}$ be the open half-plane not containing $r'$
    
      \If{$\mathcal{H} \cap \mathcal{L}_{V}(r)$ contains no robots \label{code: condition terminal} }{$result \leftarrow True$}
    
    }
    
    return $result$
    
    }
    
    
    \Fn{\textsc{ComputeDestination()} \label{code: ComputeDestination}}{

    \uIf{there are no robots in $\mathcal{H}_{R}^{O}(r)$}{
    
      return $(-1_r, 0)$
    
    }
    
    \Else{
    
	  $\mathcal{L}' \leftarrow$ the leftmost vertical line containing a robot in $\mathcal{H}_{R}^{O}(r)$
	  
	  $d \leftarrow$ the horizontal distance between $r$ and $\mathcal{L}'$
	  
	  return $(-d, 0)$
    
    }
    
    }

\caption{Phase 1 of Leader Election}
    \label{algo:phase1} 
\end{algorithm}
$\\  \\$
Suppose that a robot $r$ with $r.light =$ \texttt{off} takes a snapshot at time $t$ and finds that 
 \begin{enumerate}
  \item there are no robots in $\mathcal{H}_{L}^{C}(r)$ other than itself, i.e., $r$ is the unique leftmost robot in the configuration,
  
  \item all robots have their lights set to \texttt{off}.
 \end{enumerate}
Then we shall say that the robot $r$ \emph{finds itself eligible to become leader at time $t$}. In this case, the robot $r$ will not immediately change its light to \texttt{leader}. It will start executing \textsc{BecomeLeader()}. \textsc{BecomeLeader()} makes it move downwards (according to its local coordinate system) until there are no robots in $\mathcal{H}_{B}^{C}(r)$ other than $r$ itself. This ensures that when $r$ changes its light to \texttt{leader}, we obtain a leader configuration.

 In the case where $r$ with $r.light = $ \texttt{off} is not a unique leftmost robot, it calls the function \textsc{LeftMostTerminal()}. If \textsc{LeftMostTerminal()} returns True, $r$ will move leftwards, and otherwise it will stay put. For a robot $r$ with $r.light = $ \texttt{off}, \textsc{LeftMostTerminal()} returns True if one of the following holds.
 \begin{enumerate}
  \item There are no robots in $\mathcal{H}_{L}^{O}(r)$, i.e., it is a leftmost robot. Furthermore, it is terminal on its vertical line.  
  \item There is exactly one robot $r'$ in $\mathcal{H}_{L}^{O}(r)$ and has light set to \texttt{candidate}. Furthermore, if $\mathcal{H} \in \{\mathcal{H}_{B}^{O}(r), \mathcal{H}_{U}^{O}(r)\}$ is the open half-plane not containing $r'$, then $\mathcal{H} \cap \mathcal{L}_{V}(r)$ must not contain any robot. 
 \end{enumerate}

 The second condition is necessary in the asynchronous setting. Let $r_1$ and $r_2$ be the leftmost robots terminal on their vertical line $\mathcal{L}$ in the initial configuration. Suppose that $r_1$ takes its first snapshot earlier and hence decides to move due to the first condition. Suppose that it leaves $\mathcal{L}$ at time $t$. If $r_2$ takes a snapshot before $t$, it will decide to move. But in absence of the second condition, $r_2$ will not move if it is activated after $t$ as it is no longer a leftmost robot. So, it will become impossible to know whether $r_2$ will move future or not. The second condition ensures that both $r_1$ and $r_2$ eventually leave $\mathcal{L}$.

 \begin{figure}[h]
\centering
\subcaptionbox[Short Subcaption]{
       \label{}
}
[
    0.36\textwidth 
]
{
    \fontsize{8pt}{8pt}\selectfont
    \def\svgwidth{0.36\textwidth}
    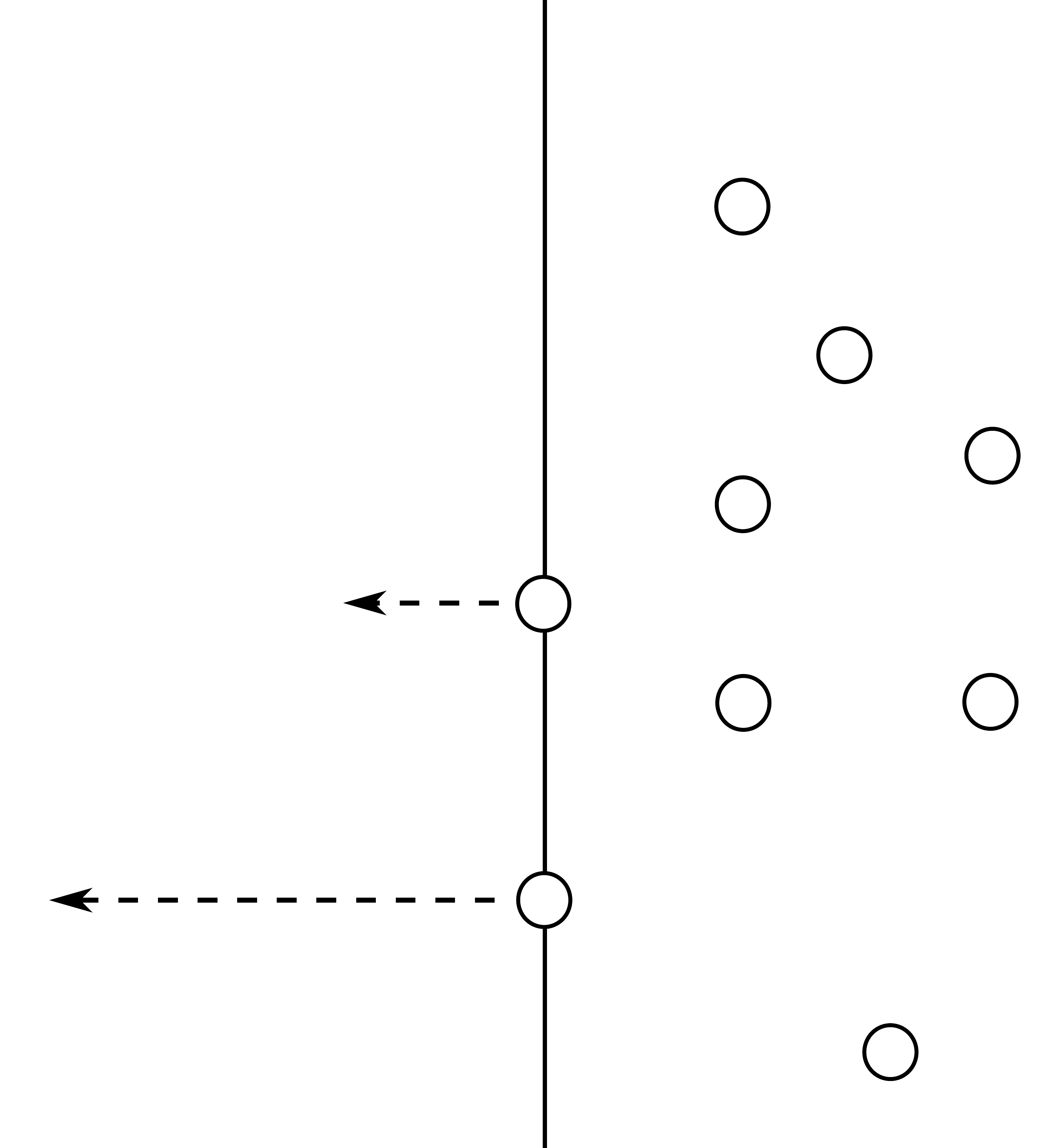
}
\hspace*{1cm}
\subcaptionbox[Short Subcaption]{
     \label{fig: compute destination 2}
}
[
    0.36\textwidth 
]
{
    \fontsize{8pt}{8pt}\selectfont
    \def\svgwidth{0.36\textwidth}
    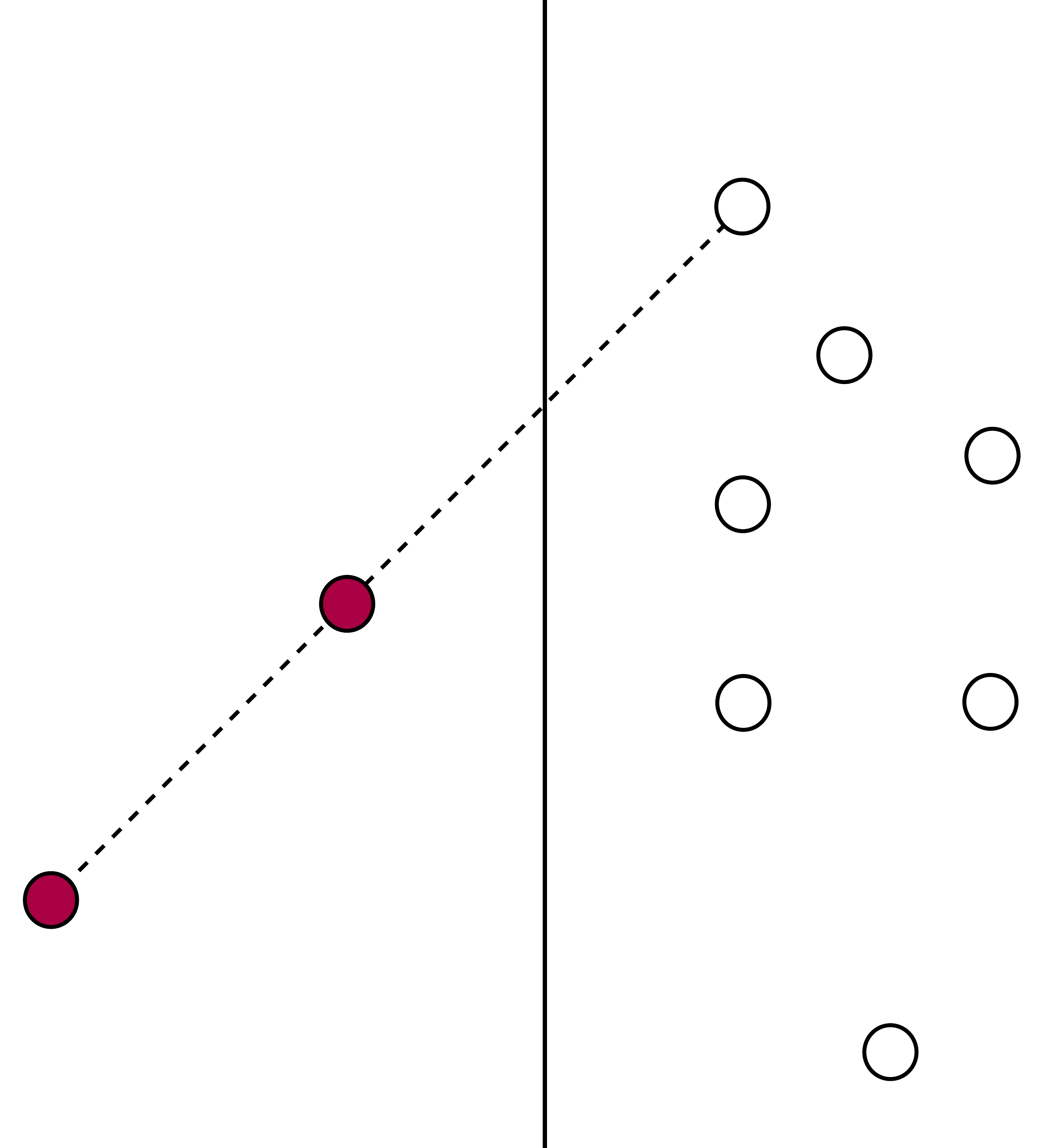
}

\caption[Short Caption]{a) The configuration has exactly two leftmost robots $r_1$ and $r_2$.  b) $r_1$ and $r_2$ move left by some arbitrary distance and change their lights to \texttt{candidate}. Then $r_3$ can see only $r_1$ on its left. }
\label{fig: compute destination}
\end{figure}

 Therefore, if the initial configuration has $k \geq 2$ leftmost robots, then both leftmost terminal robots will eventually change their lights to \texttt{terminal} and move left. After completing their moves, they will set their lights to \texttt{candidate}. Now, let us see what might happen if we do not specify how much the robots should move. Consider the situation shown in Fig. \ref{fig: compute destination}. The initial configuration $\mathbb{C}(0)$ has exactly two leftmost robots $r_1$ and $r_2$. Let $\mathbb{C}(t)$ be the configuration as shown in Fig. \ref{fig: compute destination 2}. Here $r_1$ has completed its move and changed its light to \texttt{candidate}. Clearly, \textsc{LeftMostTerminal()} will return True for $r_3$ as it can see only one robot in $\mathcal{H}_{L}^{O}(r_3)$ with light set to \texttt{candidate}. Therefore, $r_3$ will also decide to move left. To ensure that such a situation does not occur, the robots move to the points computed by the function \textsc{ComputeDestination()} (see line \ref{code: ComputeDestination} of Algorithm \ref{algo:phase1}). First consider the case where the initial configuration $\mathbb{C}(0)$ does not have all robots on the same vertical line. Let $d$ be the horizontal distance between the two leftmost vertical lines containing robots. In this case, \textsc{ComputeDestination()} makes $r_1$ and $r_2$ move left by $d$. Thus both of them move by equal amount. This ensures that the situation that we described will not occur (see the proof of Lemma \ref{lemma p1 hard}). Now consider the case where the initial configuration $\mathbb{C}(0)$ has all robots on the same vertical line. In this case, \textsc{ComputeDestination()} makes the terminal robots move by a unit distance according to their local coordinate system. Notice that their unit distances can be different. However, the situation that we described will not occur in this case even if they move by different amount (see the proof of Lemma \ref{lemma p1 line}). 
$\\  \\$
  \begin{algorithm}[H]
    \setstretch{0.1}
    \SetKwInOut{Input}{Input}
    \SetKwInOut{Output}{Output}
    \SetKwProg{Fn}{Function}{}{}
    \SetKwProg{Pr}{Procedure}{}{}

    \Input{The configuration of robots visible to me.}
    
    \Pr{\textsc{Phase2()}}{

    $r \leftarrow$ myself
    
    \uIf{$r.light =$ \texttt{candidate}}{

	  $r' \leftarrow$  another robot on $\mathcal{L} = \mathcal{L}_{V}(r)$ with light \texttt{candidate} or \texttt{symmetry}
	
	  $\mathcal{L}' \leftarrow$ the leftmost vertical line containing a robot in $\mathcal{H}_{R}^{O}(r)$
	  
	  $\mathcal{K} \leftarrow$ the horizontal line passing through the mean of the positions of the robots lying on $\mathcal{L}'$
	  
	  $d_{r} \leftarrow$ distance of $r$ from $\mathcal{K}$
	  
	  $d_{r'} \leftarrow$ distance of $r'$ from $\mathcal{K}$
	  
	  $d_{\mathcal{LL'}} \leftarrow$ distance between $\mathcal{L}$ and $\mathcal{L}'$
	  
	  set the positive direction of $Y$-axis of the local coordinate system towards $r'$ \label{code: set y axis}
    
	\uIf{ $r'.light =$  \texttt{symmetry}}{
	
	  \uIf{$d_{r} < d_{r'}$}{Move $d_{r'} - d_{r}$ distance vertically in direction opposite to $r'$}
	
	  \ElseIf{$d_{r} = d_{r'}$}{$r.light \leftarrow$ \texttt{symmetry}}

	}

	\ElseIf{$r'.light =$  \texttt{candidate}}{

	  \uIf{$r$ and $r'$ are in the same closed half-plane delimited by $\mathcal{K}$}{
	  
	    \If{$d_{r} > d_{r'}$}{Move to \textsc{ComputeDestination2()}}
	  
	  }
	  
	  \Else{
	  \uIf{$\lambda'(r) \prec \lambda'(r')$}{Move to \textsc{ComputeDestination2()}}

	  \ElseIf{$\lambda'(r) = \lambda'(r')$}{
	  
	  \If{$d_{r} \geq d_{r'}$}{

	  \uIf{the number of robots is equal to $n$ and there are no robots in $\mathcal{H}_B^O(r)$ \label{code: condition p2} }{
	  
	       \uIf{$\lambda(r) \prec \lambda(r')$}{Move to \textsc{ComputeDestination2()}}
	    
	    \uElseIf{$\lambda(r) \succ \lambda(r')$}{Move $\frac{d_{\mathcal{LL'}}}{2}$ distance right horizontally}
	    
	    \ElseIf{$\lambda(r) = \lambda(r')$}{
	    
	      \uIf{$\mathcal{K}$ has no robots on it}{Move to \textsc{ComputeDestination2()}}

	      \Else{$r.light \leftarrow$ \texttt{symmetry}}

	  }
	      
	    }

	  \Else{
	  
	    Move $d_{r}$ distance vertically in direction opposite to $r'$

	  }
	  }
	  }

	}
      }
      
    }

    \uElseIf{$r.light =$ \texttt{symmetry}}{
    
    \If{there is a robot in $\mathcal{H}_{L}^{O}(r)$ with light \texttt{off}}{
	
	  $r.light \leftarrow $ \texttt{off}
	
	}
    
    }
    
    \ElseIf{$r.light =$ \texttt{off}}{
    
    \If{there are two robots $r_1$ and $r_2$ in $\mathcal{H}_{L}^{O}(r)$ with light \texttt{symmetry} on the same vertical line $\mathcal{L}$}{
    
    $\mathcal{K} \leftarrow$ the horizontal line passing through the mid-point of the line segment joining $r_1$ and $r_2$.
    
    \If{$r$ is the leftmost robot lying on $\mathcal{K}$}{
    
    $d \leftarrow$ distance of $r$ from $\mathcal{L}$
    
    Move $d+1_r$ distance  left

    }
    
    }
    
    }

    }

    
    \Fn{\textsc{ComputeDestination2()}\label{code: ComputeDestination2}}{
    
    $\mathcal{H} \leftarrow$ the open half-plane delimited by $\mathcal{L}_{H}(r')$ not containing $r$.    
    
    \uIf{$\mathcal{L}' \cap \mathcal{H}$ has no robots}{return $(-1_r,0)$}
    
    \Else{
    
      $r'' \leftarrow$ the robot on $\mathcal{L}' \cap \mathcal{H}$ with maximum $y$-coordinate
      
      return $(-\frac{1}{2} d_{\mathcal{LL'}} \frac{r'.y}{r''.y - r'.y},0)$
    
    }
    
    }
    
\caption{Phase 2 of Leader Election}
    \label{algo:phase2} 
\end{algorithm}

\subsubsection{Phase 2}\label{sec: p2}

\paragraph{Overview.}

Phase 1 will end with either a leader configuration or a candidate configuration. In the first case, leader election is done, and in the second case we enter phase 2. So assume that we have two leftmost robots, $r_1$ and $r_2$, lying on the vertical line $\mathcal{L}$ with their lights set to \texttt{candidate}. Let $\mathbb{C}' = \mathbb{C} \setminus \{r_1,r_2\}$ be the configuration of the remaining robots. The idea is to elect the leader by inspecting the configuration $\mathbb{C}'$. If $\mathbb{C}'$ is asymmetric (Case 1), then it is possible to deterministically elect one of $r_1$ and $r_2$ as the leader from  asymmetry of $\mathbb{C}'$. The problem is that $r_1$ and $r_2$ may have only a partial view of $\mathbb{C}'$ due to obstructions. However, in that case the robot is aware that it can not see the entire configuration, as it knows the total number of robots (from the input $\mathbb{P}$). Suppose that $r_1$ finds that it can not see all the robots in the configuration. Then we can ask $r_1$ to move along the vertical line (in the direction opposite to $r_2$) in order to get an unobstructed view of the configuration. It can be proved that after finitely many steps, $r_1$ will be able to see all the robots in the configuration. However, this strategy may create a symmetry with the axis of symmetry not containing any robots, from where arbitrary pattern formation is deterministically unsolvable. This happens only in the case where $\mathbb{C}'$ has such a symmetry (Case 2). But the  $r_1$, $r_2$ can not ascertain this without having the full view of the configuration. Consider the leftmost vertical line $\mathcal{L}'$ containing a robot in $\mathbb{C}'$. Let $\mathcal{K}$ be the horizontal line passing through the mean of the positions of the robots lying on $\mathcal{L}'$. Clearly, if $\mathbb{C}'$ is symmetric, then $\mathcal{K}$ is the axis of symmetry. Note that both $r_1$, $r_2$ can see all robots lying on $\mathcal{L}'$. Hence, they can determine $\mathcal{K}$. Therefore, although they can not ascertain if $\mathbb{C}'$ is symmetric or not, they can still determine the only possible axis with respect to which there may exist a symmetry. Now $r_1$ and $r_2$ will determine their movements based on their distances from $\mathcal{K}$. In particular, only the one farther away from $\mathcal{K}$ will move, and in case of a tie, both will move. We can prove that following this strategy, it is possible to elect one among  $r_1$ and $r_2$ as the leader in Case 1 and Case 2. However, this may not be possible if $\mathbb{C}'$ is symmetric with respect to $\mathcal{K}$, with $\mathcal{K}$ containing at least one robot (Case 3). In this case, the movements of $r_1$ and $r_2$ can make the entire configuration symmetric with respect to the axis $\mathcal{K}$ (which contains at least one robot). In that case, $r_1$ and $r_2$ will change theirs lights to \texttt{symmetry}. It is crucial for the correctness of the algorithm that we coordinate the movements of $r_1$ and $r_2$ in such a way that when both set their lights to \texttt{symmetry}, they must be 1) visible to all the robots in $\mathbb{C}'$, and 2) equidistant from $\mathcal{K}$. If this is ensured, then all the robots can determine $\mathcal{K}$, which should be the horizontal line passing through the mid-point of the line segment joining $r_1$ and $r_2$. Then the leftmost robot on $\mathcal{K}$ will move towards left to become the leftmost robot in $\mathbb{C}$ and will eventually become the leader.

\paragraph{Formal Description.}

 A pseudocode description of phase 2 is presented in Algorithm \ref{algo:phase2}. Let $r_1$ and $r_2$ be the two candidate robots in the candidate configuration $\mathbb{C}$. Let $\mathbb{C}' = \mathbb{C} \setminus \{r_1,r_2\}$. As defined earlier, $\mathcal{L}$ is the vertical line containing $r_1$ and $r_2$, $\mathcal{L}'$ is the leftmost vertical line containing a robot in $\mathbb{C}'$, and $\mathcal{K}$ is the horizontal line passing through the mean of the positions of the robots lying on $\mathcal{L}'$. Let $\mathcal{H}_1$ and $\mathcal{H}_2$ be the two open half-planes delimited by $\mathcal{K}$. Let  $\overline{\mathcal{H}_1}$ and $\overline{\mathcal{H}_2}$ be their closure. First assume that both $r_1$ and $r_2$ lie in the same closed half-plane, say $r_1, r_2 \in \overline{\mathcal{H}_1}$. Then the robot further from $\mathcal{K}$, say $r_2$, will move left. Then clearly we are back to phase 1. As described in Section \ref{sec: phase 1}, eventually $r_2$ will become leader. However, this movement can create the similar situation shown in Fig. \ref{fig: compute destination 2}. To avoid this, we have specified the moving distance by the function \textsc{ComputeDestination2} (see line \ref{code: ComputeDestination2} of Algorithm \ref{algo:phase2}). See the proof of Lemma \ref{lemma: phase 2b} for details.  Now assume that $r_1$ and $r_2$ are in different open half-planes. So, let $\mathcal{H}_1$ and $\mathcal{H}_2$ be the half-planes containing $r_1$ and $r_2$ respectively. For each $\mathcal{H}_i$, we define a coordinate system $C_i$ in the following way. The point of intersection between $\mathcal{K}$ and $\mathcal{L}'$ is the origin, $\mathcal{L}' \cap \ \mathcal{H}_i$ is the positive $Y$-axis and the positive direction of $X$-axis is according to the global agreement. We express the positions of all the robots in $\mathcal{H}_i$ with respect to the coordinate system $C_i$. Now arrange the positions in the lexicographic order. Let $\lambda(r_i)$ denote the string thus obtained (see Fig. \ref{fig: lex}). Each term of the string is an element from $\mathbb{R}^2$. To make the length of the strings $\lambda(r_1)$ and $\lambda(r_2)$ equal, null elements $\Phi$ may be appended to the shorter string. For any non-null term $(x,y)$ of a string, we set $(x,y) < \Phi$. We shall write $\lambda(r_1) \prec \lambda(r_2)$ iff $\lambda(r_2)$ is lexicographically larger than $\lambda(r_1)$. For each string $\lambda(r_i)$, let $\lambda'(r_i)$ be the string obtained from $\lambda(r_i)$ by deleting all terms with $x$-coordinate not equal to $0$. That is, the terms of $\lambda'(r_i)$ corresponds to the robots on $\mathcal{L}' \cap \ \mathcal{H}_i$.  Again, null elements $\Phi$ may be appended to make the length of the strings $\lambda'(r_1)$ and $\lambda'(r_2)$ equal. The plan is to choose a leader by comparing these strings. First, the robots will compare $\lambda'(r_1)$ and $\lambda'(r_2)$. Clearly, both of them can see all the robots on $\mathcal{L}'$, hence can compute $\lambda'(r_1)$ and $\lambda'(r_2)$. If $\lambda'(r_i) \prec \lambda'(r_j)$, $r_i$ will move left. As before, $r_i$ will become leader. If $\lambda'(r_1) = \lambda'(r_2)$, the robots have to compare the full strings $\lambda(r_1)$ and $\lambda(r_2)$. But in order to compute these strings, the complete view of the configuration is required. Let $d_{r_1}$ and $d_{r_2}$ be the distance of $r_1$ and $r_2$ from $\mathcal{K}$ respectively. If $d_{r_i} \geq d_{r_j}$, $r_i$ will move vertically away from $r_j$  by a distance $d_{r_i}$ to get the full view of the configuration. When it can see all the robots in the configuration, it computes $\lambda(r_1)$ and $\lambda(r_2)$. If $\lambda(r_i) \prec \lambda(r_j)$, $r_i$ will move left ($r_i$ will become leader) and if $\lambda(r_i) \succ \lambda(r_j)$, $r_i$ will move right ($r_j$ will become leader). In the case where $\lambda(r_i) = \lambda(r_j)$ ($\mathbb{C}'$ is symmetric), $r_i$ will move left if $\mathcal{K}$ has no robots on it, or otherwise it will change its light to \texttt{symmetry}.

\begin{figure}[h]
\centering
    \fontsize{8pt}{8pt}\selectfont
    \def\svgwidth{0.37\textwidth}
    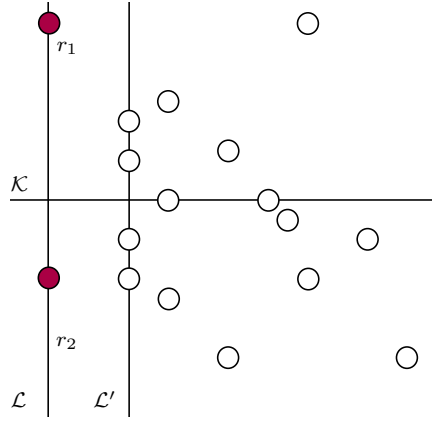

\caption{There are in total 15 robots in $\mathbb{C}'$ with 5 of  them in $\mathcal{H}_1$, 8 in $\mathcal{H}_2$ and 2 on $\mathcal{K}$. Hence, the length of $\lambda(r_1)$ and $\lambda(r_2)$ will be 8, with $\lambda(r_1)$ having 3 null element terms. With respect to some unit distance, we have $\lambda(r_1) =$ $(0,1)(0,2)(1,2.5)(2.5,1.25)(4.5,4.5)\Phi\Phi\Phi$ and $\lambda(r_2) =$ $(0,1)(0,2)(1,2.5)(2.5,4)(4,0.5)(4.5,2)(6,1)(7,4)$. So, $\lambda(r_1) \prec \lambda(r_2)$. Also, $\lambda'(r_1) =$ $\lambda'(r_2) =$ $(0,1)(0,2)$.}
\label{fig: lex}
\end{figure}

\subsection{Pattern Formation from a Leader Configuration}\label{sec:election}

\paragraph{Overview.}

In a leader configuration, all non-leader robots lie on one of the open half-planes delimited by the horizontal line passing through the leader $r_l$. This leads to an agreement on the direction of $Y$ axis, as we can set the empty open half-plane to correspond to the negative direction of $Y$-axis or `down'. Hence, we have an agreement on `up', `down', `left' and `right'. However, we still do not have a common notion of unit distance. Notice that all the non-leader robots are in $\mathcal{H}_R^O(r_l)$. Now, we shall ask one of the  non-leader robots to move to the line $\mathcal{L}_{V}(r_l)$. The distance of this robot from the leader $r_l$ will be set as the unit distance. Now it only remains to fix the origin. We shall set the origin at a point such that the coordinates of $r_l$ are $(-1,-1)$. Now that we have a common fixed coordinate system, we embed the pattern on the plane. Let us call these points the \emph{target points}. We have to bring the robots to these points without causing any collisions. Take the projection of the target points on $\mathcal{L}_{H}(r_l)$. Our plan is to first sequentially bring the robots to these points and then sequentially move them to the corresponding target points. The problem is that there might be multiple target points on the same vertical line and thus we may have one projected point corresponding to multiple target points. Therefore, to accommodate multiple robots, we shall assign pockets of space of a fixed size around each projected point.

\paragraph{Formal Description.}

A pseudocode description of this stage is given in Algorithm \ref{main_algorithm: stage 2}. At the start of this stage, we have a leader configuration. Let $r_l$ be the leader having its light set to \texttt{leader}. Any robot that can see $r_l$ starts executing Algorithm \ref{main_algorithm: stage 2}.
$\\ \\$
  \begin{algorithm}[H]
    \setstretch{.1}
    \SetKwInOut{Input}{Input}
    \SetKwInOut{Output}{Output}
    \SetKwProg{Fn}{Function}{}{}
    \SetKwProg{Pr}{Procedure}{}{}

  \Input{The configuration of robots visible to me.}
    
    \Pr{\textsc{PatternFormationFromLeaderConfiguration()}}{

    $r \leftarrow$ myself
    
    $r_l \leftarrow$ the robot with light \texttt{leader}
    
    \uIf{$r.light =$ \texttt{off}}{
    
      \uIf{$r_l \in \mathcal{H}_B^O(r) \cap \mathcal{H}_L^O(r)$}{
	
      \uIf{there is a robot with light \texttt{done}}{$r.light \leftarrow$ \texttt{done} \label{code: unit off to done} }
      
                
      \ElseIf{there are no robots in $\mathcal{H}_B^O(r) \cap \mathcal{H}_U^O(r_l) \cap \mathcal{H}_R^O(r_l)$ and $r$ is the leftmost robot on $\mathcal{L}_H(r) \cap \mathcal{H}_R^O(r_l)$ \label{code: condition stage2}}{
      
	\uIf{there is no robot on $\mathcal{L}_V(r_l)$ \label{code: unit condition}}{
	
	    $p \leftarrow$ the point of intersection of the lines $\mathcal{L}_H(r)$ and $\mathcal{L}_V(r_l)$
	
	    Move to $p$
	}
	\Else{
	
	\If{there are $k \geq 0$ robots on $\mathcal{L}_H(r_l)$ with light \texttt{off}}{Move to $(\Psi(k+2), -1)$}}

%
%
%
	  
	}
    
      }
      
      \uElseIf{$r_l \in \mathcal{L}_H(r)$, there is a robot on $\mathcal{L}_V(r_l)$, and there are no robots with light \texttt{off} in  $\mathcal{H}_U^O(r) \cap \mathcal{H}_R^O(r_l)$ \label{code: final condition1}}{
      
      \If{$r.pos =$ $(\Psi(i), -1)$ and \textsc{PartialFormation}$(i)$ = True \label{code: final condition1b}}{
      
      $r.light \leftarrow$ \texttt{done}
      
      Move to $\mathbb{P}[i]$}

      }

      \ElseIf{$r_l \in \mathcal{L}_V(r)$}{

    \If{there is no other robot with light \texttt{off} in $\mathcal{H}_U^C(r)$ and there are no robots in $\mathcal{H}_B^O(r)$ except $r_l$ \label{code: final condition3}}{
    
    
    Move to $\mathbb{P}[1]$}
    
    }

    }

    \ElseIf{$r.light =$ \texttt{leader}}{
    
      \If{there are no robots with light \texttt{off}}{
    
    $r.light \leftarrow$ \texttt{done}
    
    Move to $\mathbb{P}[0]$}
    
      }
    }

    \Pr{\textsc{PartialFormation}$(i)$}{
    
    \uIf{$i = 2$}{return True}
    
     \Else{
    
      \uIf{there is a robot with light \texttt{done} at $s_{i-1}$ \label{code: final condition2}}{return True}
      
      \Else{return False}

    }

    }

\caption{Pattern Formation from Leader Configuration}
    \label{main_algorithm: stage 2} 
\end{algorithm}
$\\ \\$
Initially all robots other than $r_l$ are inside the region $\mathcal{H}_U^O(r_l) \cap \mathcal{H}_R^O(r_l)$. Order the robots in $\mathcal{R} \setminus \{r_l\}$ from bottom to up, and from left to right in case multiple robots on the same horizontal line (see Fig. \ref{fig: stage2_robot_order}). Let us denote these robots by $r_1, r_2, \dots, r_{n-1}$ in this order . The robots will move sequentially according to this order. The robot $r_1$ will move to $\mathcal{L}_V(r_l)$, while the robots $r_2, \ldots, r_{n-1}$ will move to $\mathcal{L}_H(r_l)$. Each $r_i,~ i = 1, \ldots, n-1$ will decide to move when it finds that 1) there are no robots in $\mathcal{H}_B^O(r_i) \cap \mathcal{H}_U^O(r_l) \cap \mathcal{H}_R^O(r_l)$, and 2) $r_i$ is the leftmost robot on $\mathcal{L}_H(r_i) \cap \mathcal{H}_R^O(r_l)$. These conditions will ensure that the robots execute their moves sequentially according to ordering in Fig. \ref{fig: stage2_robot_order}. First $r_1$ will move horizontally to $\mathcal{L}_V(r_l)$. After $r_1$ reaches $\mathcal{L}_V(r_l)$, we shall denote it by $r_u$. Any robot that can see both $r_l$ and $r_u$, sets its coordinate system as shown in Fig. \ref{fig: stage2_robot_agree} so that $r_l$ is at $(-1,-1)$ and $r_u$ is at $(-1,0)$. We shall refer to this coordinate system as the \emph{agreed coordinate system}. Now embed the pattern $\mathbb{P}$ in the plane with respect to this coordinate system. Let $s_i$ denote the point in plane corresponding to $\mathbb{P}[i]$. Let us call $s_0, s_1, \ldots, s_{n-1}$ the target points. Recall that each $\mathbb{P}[i]$ is from $\mathbb{R}_{\geq 0}^2$. Also, $\mathbb{P}$ is sorted in lexicographic order. Therefore, $\{s_0, s_1, \ldots, s_{n-1}\} \subset \mathcal{H}_U^O(r_l) \cap \mathcal{H}_R^O(r_l)$ and are ordered from left to right, and from bottom to up in case there are multiple robots on the same vertical line (see Fig. \ref{fig: stage2_target0}). For each $s_i$, $i = 0, 1, \ldots, n-1$, we define a point $p_i = (\Psi(i), -1)$ on $\mathcal{L}_H(r_l)$ in the following way. Let $s_i = (x_i, y_i)$. Let $\mathcal{L}_i$ be the vertical line passing through $s_i$, i.e., the line $X = x_i$. Let $s_i$ be the $k$th target point on $\mathcal{L}_i$ from bottom, i.e., the $k$th item in $\mathbb{P}$ with the first coordinate equal to $x_i$. Then $\Psi(i) = x_i + \frac{(k-1)}{2(m_i-1)}\epsilon$, where $m_i$ is equal to the total number of target points on $\mathcal{L}_i$, and $\epsilon$ is equal to the smallest horizontal distance between any two target points not on the same vertical line, or equal to 1 if all target points are on the same vertical line (see Fig. \ref{fig: stage2_target}). Notice that the numbers $\Psi(0), \Psi(1), \ldots, \Psi(n-1)$ are calculated from $\mathbb{P}$ alone.  The robots  $r_2, \dots, r_{n-1}$ will sequentially move to $p_2, \dots, p_{n-1}$ respectively (see Fig. \ref{fig: stage2_robot_move}). When an $r_i$ ($i = 2, \ldots, n-1$) has to move, it can see $i-2$ robots with light \texttt{off} on $\mathcal{L}_H(r_l)$ and thus, decides to move to $p_i$. Clearly, $r_i$ can see both $r_l$ and $r_u$ and therefore, can determine the point $p_i = (\Psi(i), -1)$.
 
 Once these moves are completed, we have $r_2, \dots, r_{n-1}$ on $\mathcal{L}_H(r_l)$ at $p_2, \dots, p_{n-1}$ respectively. Now $r_2, \dots, r_{n-1}$ will sequentially move to $s_2, \dots, s_{n-1}$ respectively (see Fig. \ref{fig: stage2_final}). To make the robots move in this order, we ask a robot to move when it can see $r_l$ on the same horizontal line and can also see $r_u$. But now the robot has to decide to which target point it is supposed to go. Since it can see both $r_l$ and $r_u$, it can compute its coordinates in the agreed coordinate system. If it finds its coordinates to be equal to $(\Psi(i), -1)$, then it decides that it has to move to $s_i$. But if $i > 2$, it will wait till $r_2, \ldots, r_{i-1}$ complete their moves. Hence, it will move only when it sees a robot with light \texttt{done} at $s_{i-1}$. When all these conditions are satisfied, it will change its light to \texttt{done} and move to its destination.  

 Now $r_u$ and $r_l$ have to move to $s_1$ and $s_0$ respectively. $r_u$ moves when it finds that there is no other robot with light \texttt{off} in $\mathcal{H}_U^C(r_u)$ and there are no robots in $\mathcal{H}_B^O(r_u)$ except $r_l$. However, in this case, it will change its light to \texttt{done} after the move (see line \ref{code: unit off to done} in Algorithm \ref{main_algorithm: stage 2}). When $r_l$ sees no robots with light \texttt{off}, it decides to move to $s_0$. However, since there is no longer a robot on $\mathcal{L}_V(r_l)$ , it can not ascertain the unit distance of the agreed coordinate system and therefore, can not determine the point $s_0$ in the plane. Hence, it has to locate the point $s_0$ by some other means. Let $r$ be the leftmost (and bottommost in case of tie) robot that $r_l$ can see. Clearly, $r$ is at $s_1$. Hence, $r_l$ knows the point in the plane with coordinates $\mathbb{P}[1]$, and also point with coordinates $(-1,-1)$ (its own position). From these two points, it can easily determine the point $s_0$, i.e., the point in the plane with coordinates $\mathbb{P}[0]$ in the agreed coordinate system. Therefore, it will change its light to \texttt{done} and move to $s_0$.


\section{Correctness and the Main Results}\label{sec: main result and proof}

\subsection{Correctness of Algorithm \ref{main_algorithm}}

\begin{lemma}\label{lemma eligible to leader}
If at some time $t$
 \begin{enumerate}
  \item a robot $r$ with light set to \texttt{off} finds itself eligible to become leader, and
  
  \item for all $r' \in$ $\mathcal{R} \setminus \{r\}$, $r'$ is stable and $r'.light =$ \texttt{off},
  \end{enumerate}
then $\exists~\ t' > t$ such that $\mathbb{C}(t')$ is a leader configuration with leader $r$.
\end{lemma}

\begin{proof}
 Clearly, $r$ will start executing \textsc{BecomeLeader()}. The robot $r$ will move downwards according to its local coordinate system if it finds any robot in $\mathcal{H}_B^C(r)$ other than itself. Clearly, after finitely many moves, it will find that $\mathcal{H}_B^C(r)$ has no robots other than itself. Notice that no livelock is created as its local coordinate system (and hence, its perception of up and down) is unchanged in each LCM cycle. So, $r$ will change its light to \texttt{leader} at some time $t' > t$ and $\mathbb{C}(t')$ will be a leader configuration. 
\end{proof}

  \begin{figure}[h!]
\centering
\subcaptionbox[Short Subcaption]{
       \label{fig: stage2_robot_order}
}
[
    0.45\textwidth 
]
{
    \fontsize{8pt}{8pt}\selectfont
    \def\svgwidth{0.45\textwidth}
    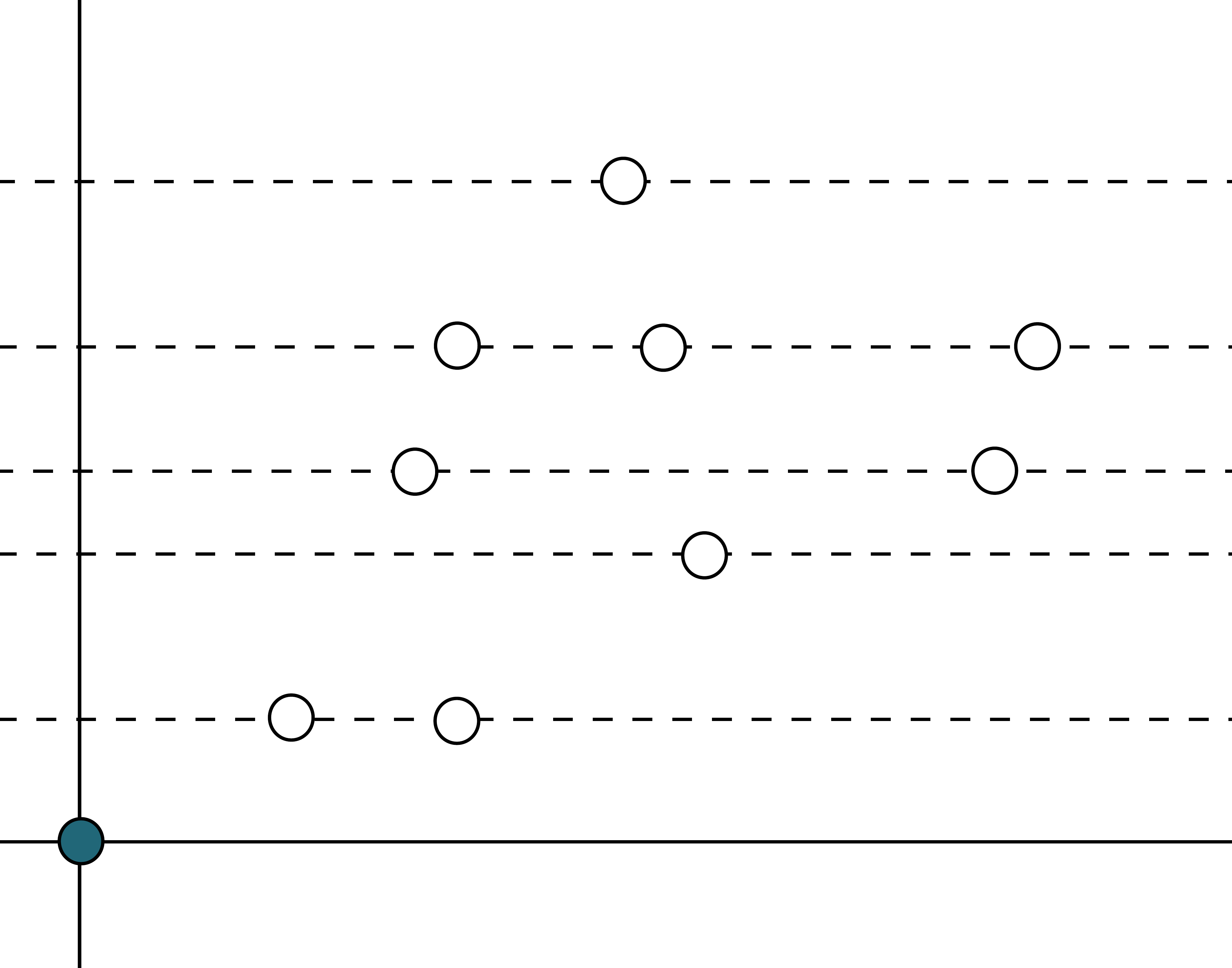
}
\hspace*{1cm}
\subcaptionbox[Short Subcaption]{
     \label{fig: stage2_robot_agree}
}
[
    0.45\textwidth 
]
{
    \fontsize{8pt}{8pt}\selectfont
    \def\svgwidth{0.45\textwidth}
    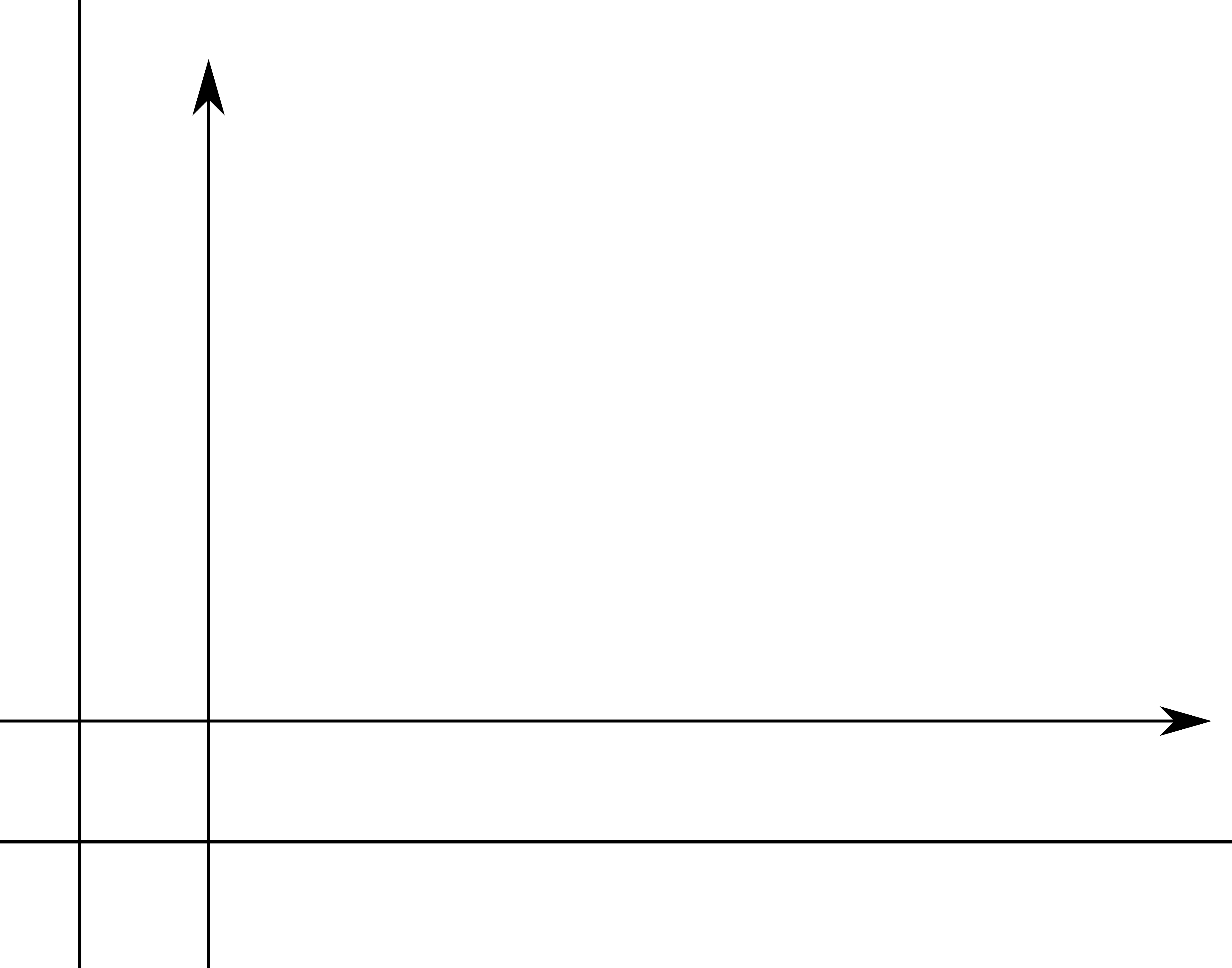
}
\hfill
\subcaptionbox[Short Subcaption]{
     \label{fig: stage2_target0}
}
[
    0.45\textwidth 
]
{
    \fontsize{8pt}{8pt}\selectfont
    \def\svgwidth{0.45\textwidth}
    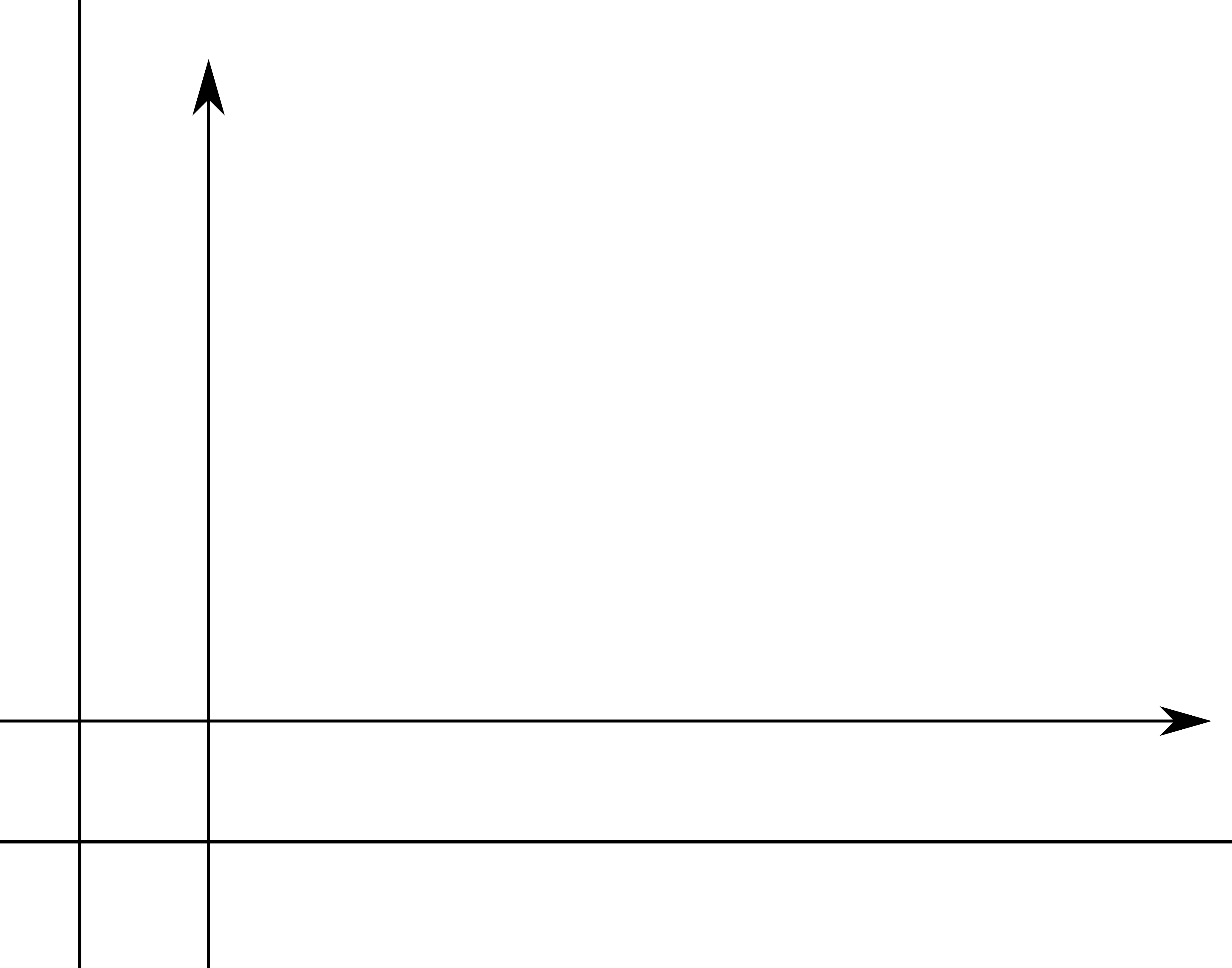
}
\hspace*{1cm}
\subcaptionbox[Short Subcaption]{
     \label{fig: stage2_target}
}
[
    0.45\textwidth 
]
{
    \fontsize{8pt}{8pt}\selectfont
    \def\svgwidth{0.45\textwidth}
    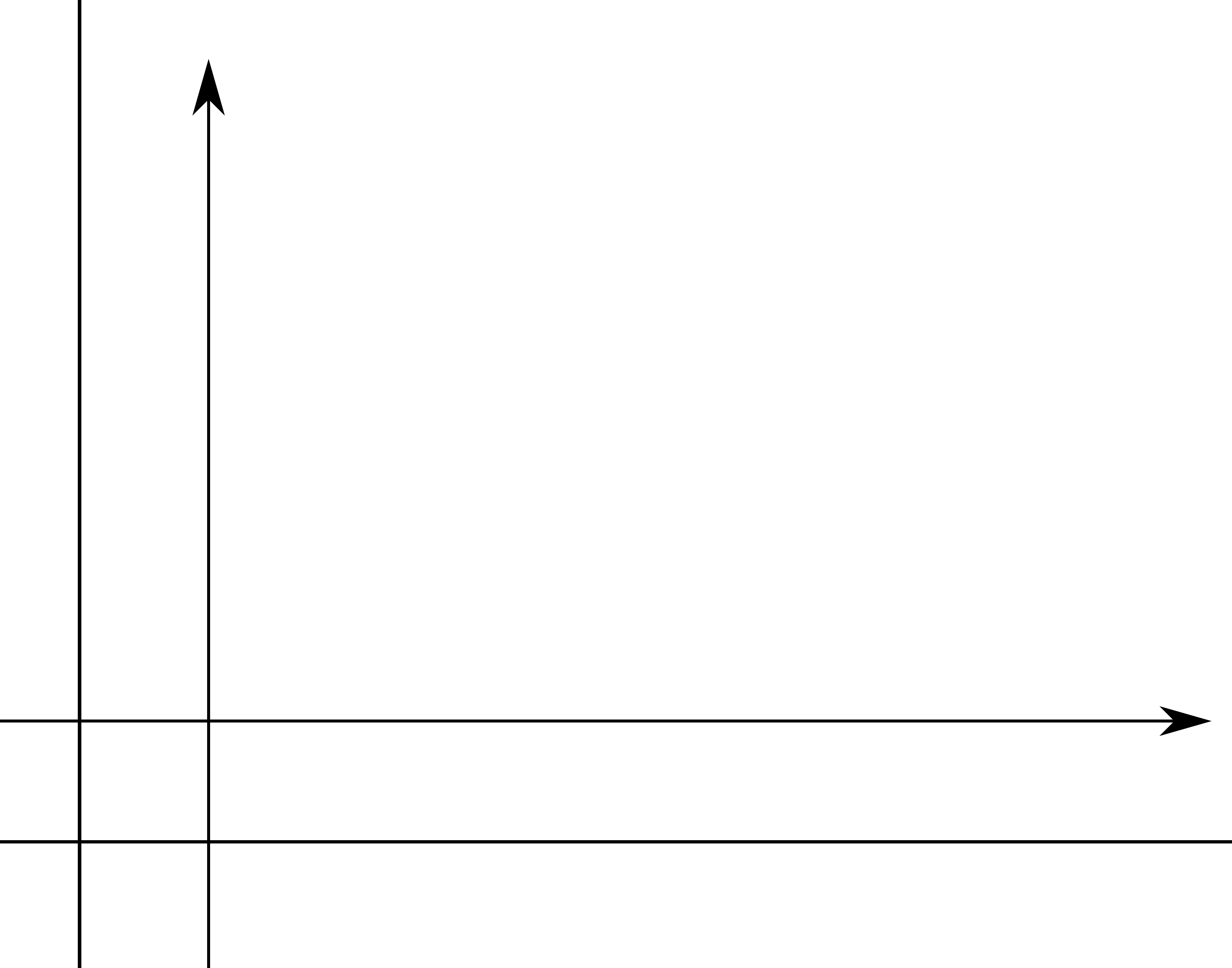
}
\hfill
\subcaptionbox[Short Subcaption]{
     \label{fig: stage2_robot_move}
}
[
    0.45\textwidth 
]
{
    \fontsize{8pt}{8pt}\selectfont
    \def\svgwidth{0.45\textwidth}
    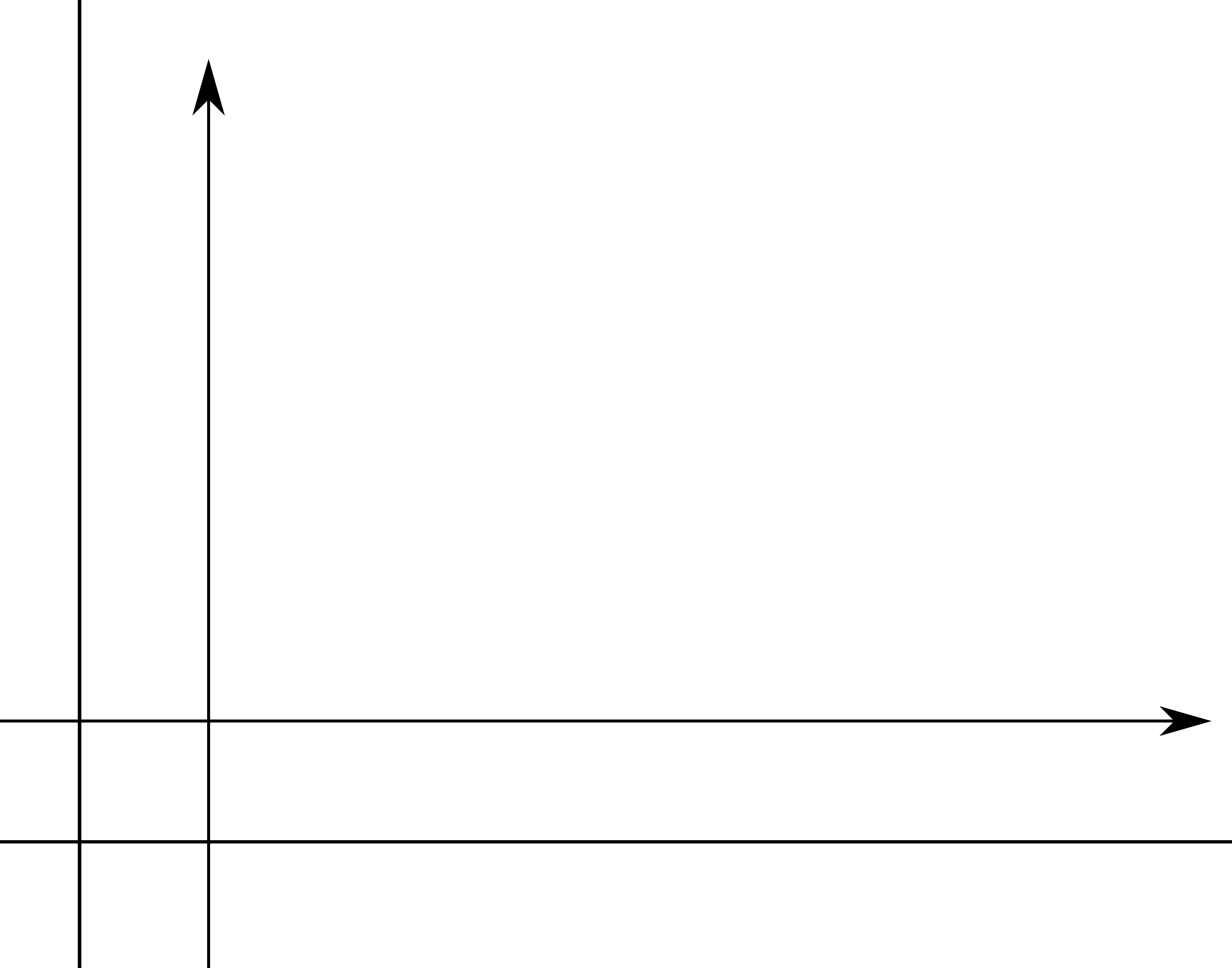
}
\hspace*{1cm}
\subcaptionbox[Short Subcaption]{
     \label{fig: stage2_final}
}
[
    0.45\textwidth 
]
{
    \fontsize{8pt}{8pt}\selectfont
    \def\svgwidth{0.45\textwidth}
    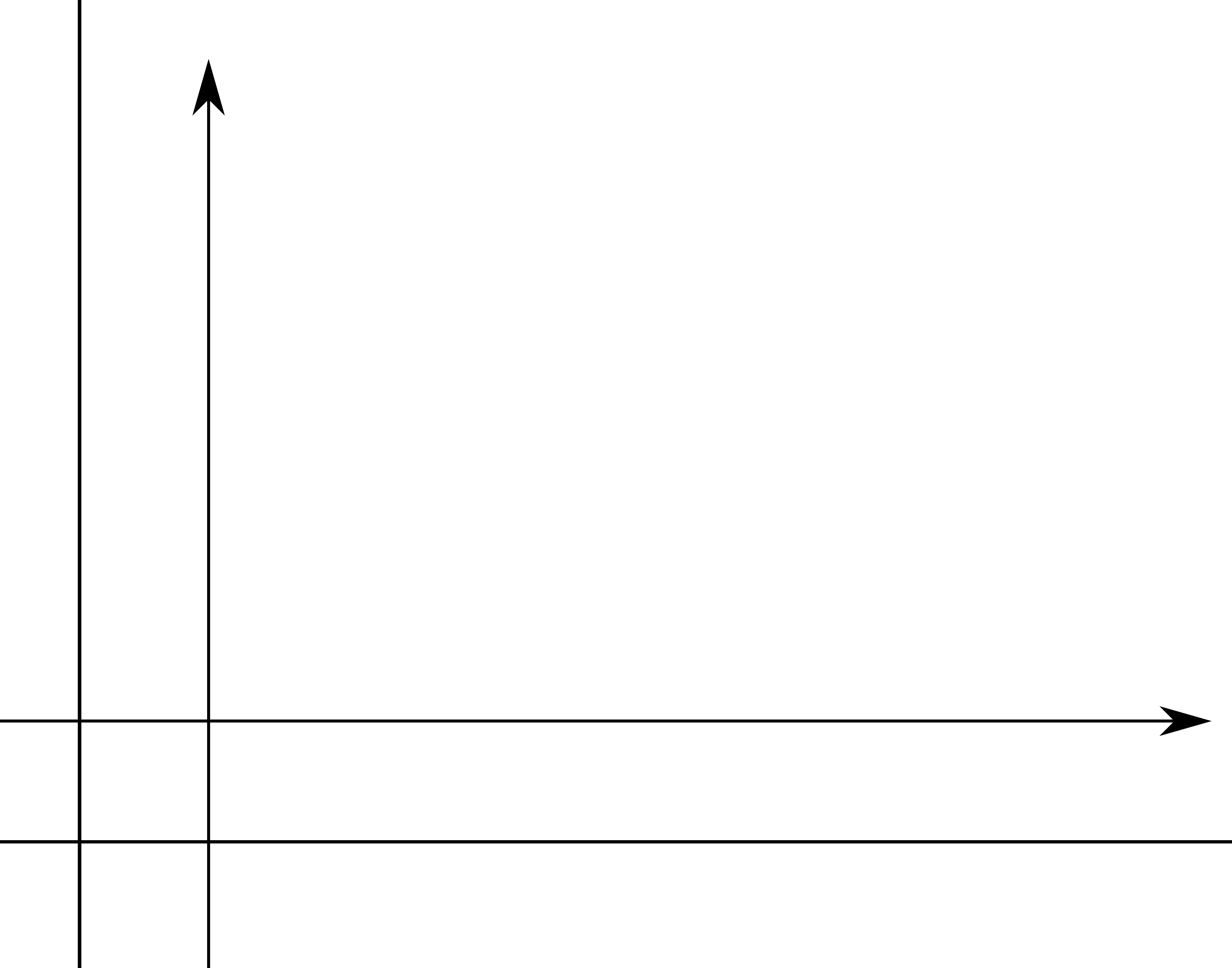
}
\caption[Short Caption]{a) A leader configuration with leader $r_l$ with its light set to \texttt{leader}. b) The configuration with $r_l$ and $r_u$. A  coordinate system is set so that the coordinates of $r_l$ and $r_u$ are $(-1,-1)$ and $(-1,0)$ respectively. c) The given pattern $\mathbb{P}$ is embedded in the plane with respect to the agreed coordinate system. d) $p_i = (\Psi(i), -1)$ corresponding to each target point $s_i$. e) Movement of the robots to $p_2, \ldots, p_{n-1}$. f) Movement of the robots to $s_0, \ldots, s_{n-1}$. }
\label{}
\end{figure}

\begin{lemma} \label{lemma p1 easy}
 If the initial configuration $\mathbb{C}(0)$ has a unique leftmost robot, then $\exists~\ T_1 > 0$ such that $\mathbb{C}(T_1)$ is a leader configuration.
\end{lemma}

\begin{proof}
 Follows from Lemma \ref{lemma eligible to leader}. 
\end{proof}

\begin{lemma}\label{lemma p1 hard}
 If the initial configuration $\mathbb{C}(0)$ has $k$ leftmost robots where $1 < k < n$, then $\exists~\ T_1 > 0$ such that $\mathbb{C}(T_1)$ is  a candidate configuration.
\end{lemma}

\begin{proof}

Let $\mathcal{L}$ denote the vertical line passing through the leftmost robots in $\mathbb{C}(0)$. Let $r_1$ and $r_2$ be the terminal robots on $\mathcal{L}$. Let $\mathcal{R} = \{r_1, r_2\} \cup \mathcal{R}' \cup \mathcal{R}''$, where $\mathcal{R}'$ is the set of non-terminal robots on $\mathcal{L}$ and $\mathcal{R}''$ is the set of robots not lying on $\mathcal{L}$. Clearly $\mathcal{R}'' \neq \emptyset$, as $k < n$.
Let $\mathcal{L}'$ denote the vertical line passing through the leftmost robots in $\mathcal{R}''$. Let $d$ be the horizontal separation between $\mathcal{L}$ and $\mathcal{L}'$.

Clearly, any robot in $\mathcal{R}' \cup \mathcal{R}''$ taking a snapshot in this configuration, will not decide to move or change its light.  Let $r_1$ and $r_2$ take their first snapshot at $t_1$ and $t_2$ respectively. Without loss of generality, we assume that $t_1 \leq t_2$. Then $r_1$ will call the function \textsc{LeftMostTerminal()} which should return True. So, $r_1$ will change its light to \texttt{terminal}, leave $\mathcal{L}$ (start moving at time $t'_1 > t_1$), move $d$ distance leftwards, and change its light to \texttt{candidate} (at time $t''_1 > t'_1$) in the next LCM cycle.

\textbf{Case 1:} Let $t_2 \in [t_1, t'_1]$. Clearly all robots in $\mathcal{R}' \cup \mathcal{R}''$ are stable and $r_1$ is stationary in $[0,t_2]$. Hence, $r_2$ will find itself to be terminal on $\mathcal{L}$ and $\mathcal{H}_{L}^{O}(r_2)$ to be empty. So, $r_2$ will also change its light to \texttt{terminal}, leave $\mathcal{L}$, move $d$ distance leftwards and change its light to \texttt{candidate} (at some time $t''_2 > t_2$).

The robot that changes its light to \texttt{candidate} first, will be stable till $T_1 =$ max$\{t''_1,t''_2\}$. This is because a robot with light \texttt{candidate}, decides to execute a non-null action only if it finds another robot with light \texttt{candidate} or finds a robot on its left. So, it only remains to show that all robots in $\mathcal{R}' \cup \mathcal{R}''$ are stable in $[t_2, T_1]$. For a robot $r \in \mathcal{R}'$, \textsc{LeftMostTerminal()} returns False if both $r_1$ and $r_2$ are on $\mathcal{L}$ as $r$ is then non-terminal. When both $r_1$ and $r_2$ are not on $\mathcal{L}$, $r$ finds two robots in $\mathcal{H}_{L}^{O}(r)$ and hence, \textsc{LeftMostTerminal()} will return False. The only situation where \textsc{LeftMostTerminal()} may return True is when only one of $r_1$ and $r_2$, say $r_1$, is in $\mathcal{H}_{L}^{O}(r)$ with its light set to \texttt{candidate}. But then clearly the condition in line \ref{code: condition terminal} in Algorithm \ref{algo:phase1} is not satisfied for $r$, as $r_2$ is still on  $\mathcal{L}$. So, \textsc{LeftMostTerminal()} will return False for $r$. Now consider $r \in \mathcal{R}''$. If $\mathcal{R}' \neq \emptyset$, then $r$ must see at least one robot with light set to \texttt{off} on its left, and  \textsc{LeftMostTerminal()} will return False. So now consider the case with  $\mathcal{R}' = \emptyset$ where there is a possibility that the situation shown in Fig. \ref{fig: compute destination 2} might arise. We shall show that this can not happen. Let $\mathcal{L}'$ be the leftmost vertical line containing a robot in $\mathcal{R}''$. Clearly, if $r$ is not on $\mathcal{L}'$, then it can see at least one robot with light set to \texttt{off}. So, let $r$ be on $\mathcal{L}'$. If it can see both $r_1$ and $r_2$ in $\mathcal{H}_{L}^{O}(r)$, then \textsc{LeftMostTerminal()} returns False. So, assume that $r$ can see only one of them, say $r_1$, in $\mathcal{H}_{L}^{O}(r)$. Hence, $r_2$ must be obstructed by $r_1$ and $r_2 \in \mathcal{H}_{L}^{O}(r_1)$. But as both $r_1$ and $r_2$ have decided to move the same distance $d$ (= the horizontal distance between $\mathcal{L}$ and $\mathcal{L}'$), the destination of $r_1$ and $r_2$ lie on the same vertical line. This implies that when $r$ takes the snapshot, either $r_1$ is moving or is yet to move. Hence, the light of $r_1$ must be set to either \texttt{terminal} or \texttt{off} (see Fig. \ref{fig: proof A}). Therefore, \textsc{LeftMostTerminal()} will return False for $r$. Thus, we have shown that all the robots in $\mathcal{R} \setminus \{r_1,r_2\}$ are stable throughout $[0,T_1]$. In particular, at $T_1$, all robots in $\mathcal{R} \setminus \{r_1,r_2\}$ are stable. Also at $T_1$, $r_1$ and $r_2$ are the two leftmost robots, both stable and have their lights set to \texttt{candidate}. Therefore, $\mathbb{C}(T_1)$ is a candidate configuration.

 \begin{figure}[!h]
\centering
\subcaptionbox[Short Subcaption]{
       \label{}
}
[
    0.3\textwidth 
]
{
    \fontsize{8pt}{8pt}\selectfont
    \def\svgwidth{0.3\textwidth}
    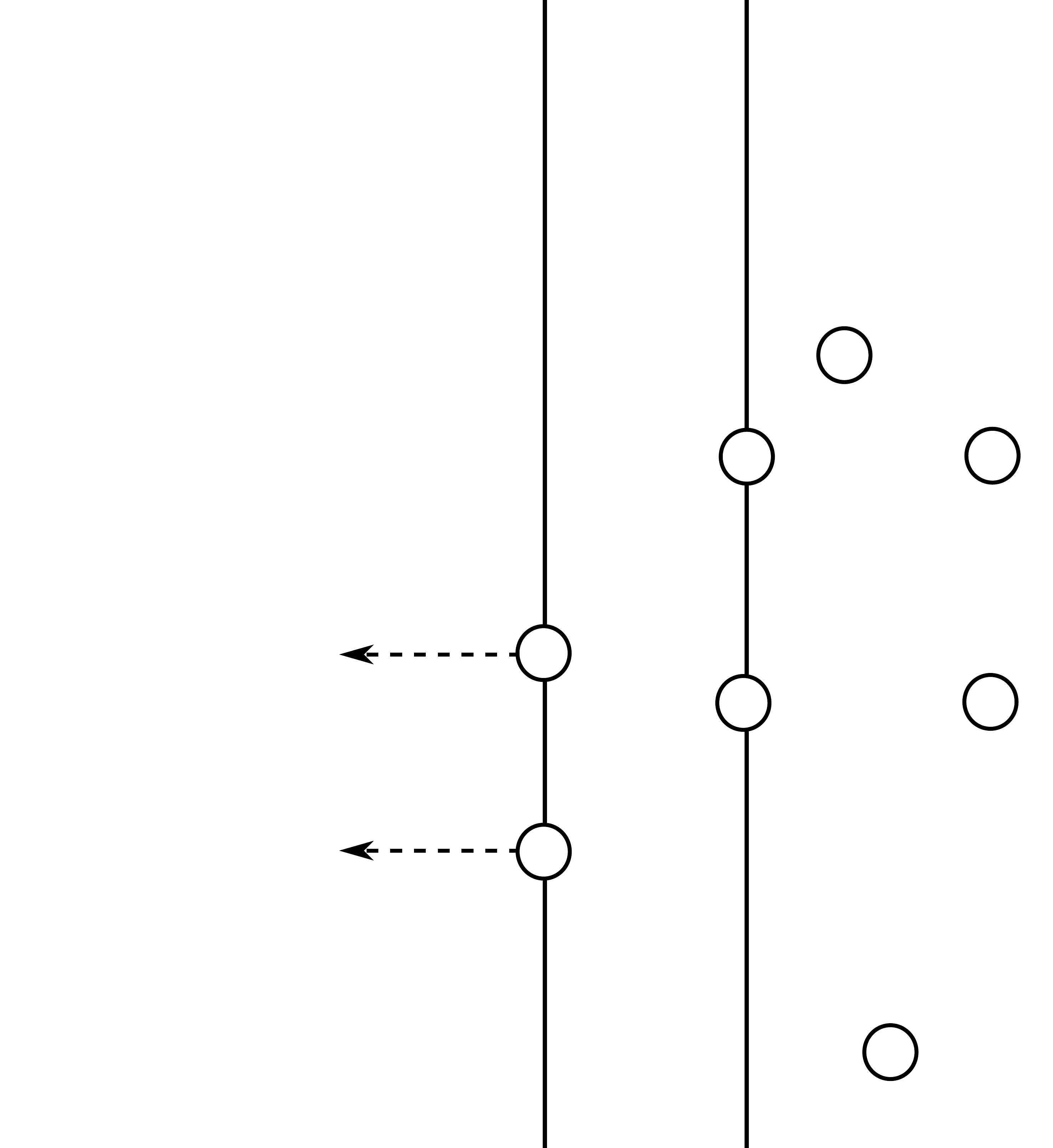
}
\hfill
\subcaptionbox[Short Subcaption]{
     \label{}
}
[
    0.3\textwidth 
]
{
    \fontsize{8pt}{8pt}\selectfont
    \def\svgwidth{0.3\textwidth}
    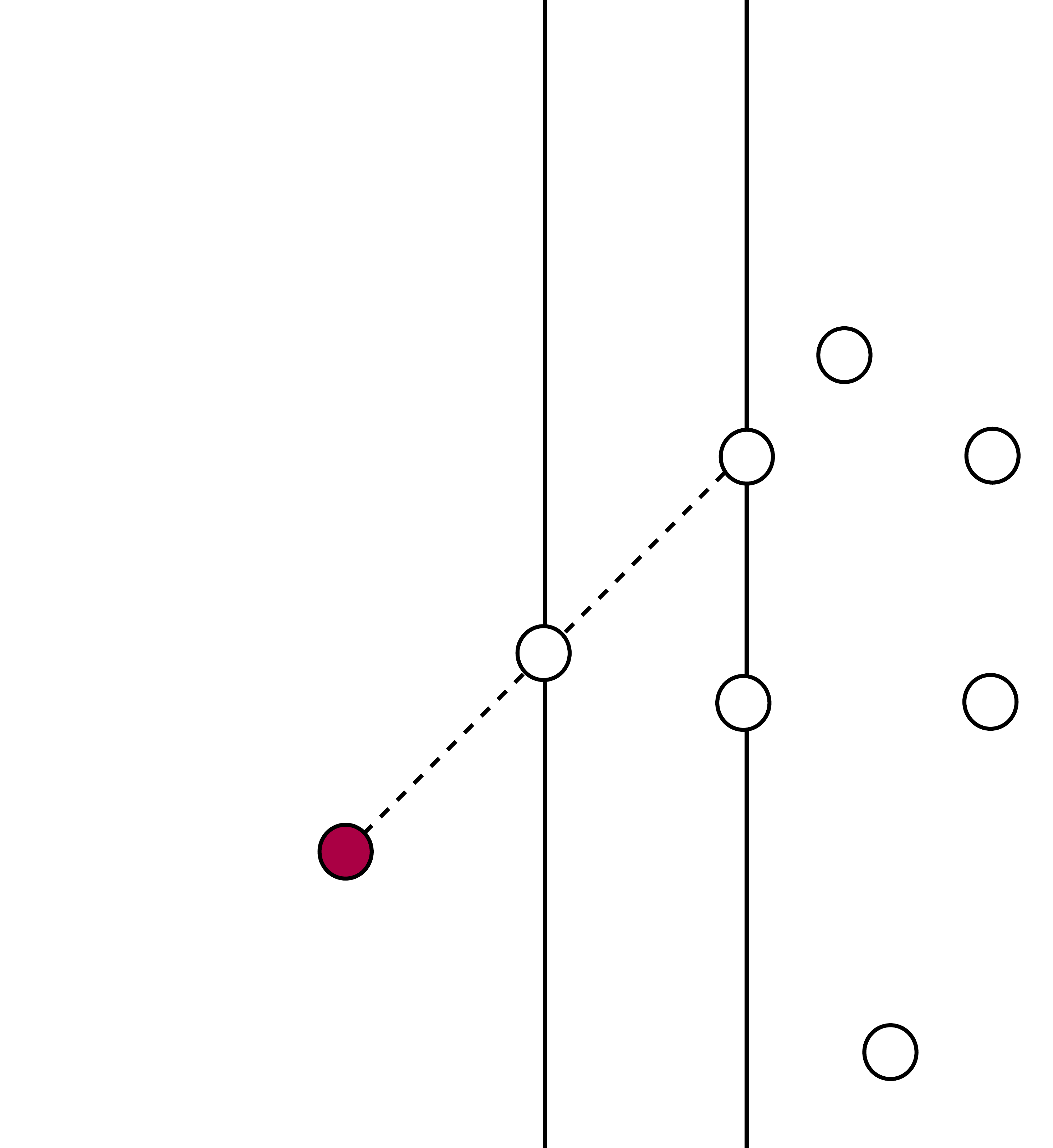
}
\hfill
\subcaptionbox[Short Subcaption]{
       \label{}
}
[
    0.3\textwidth 
]
{
    \fontsize{8pt}{8pt}\selectfont
    \def\svgwidth{0.3\textwidth}
    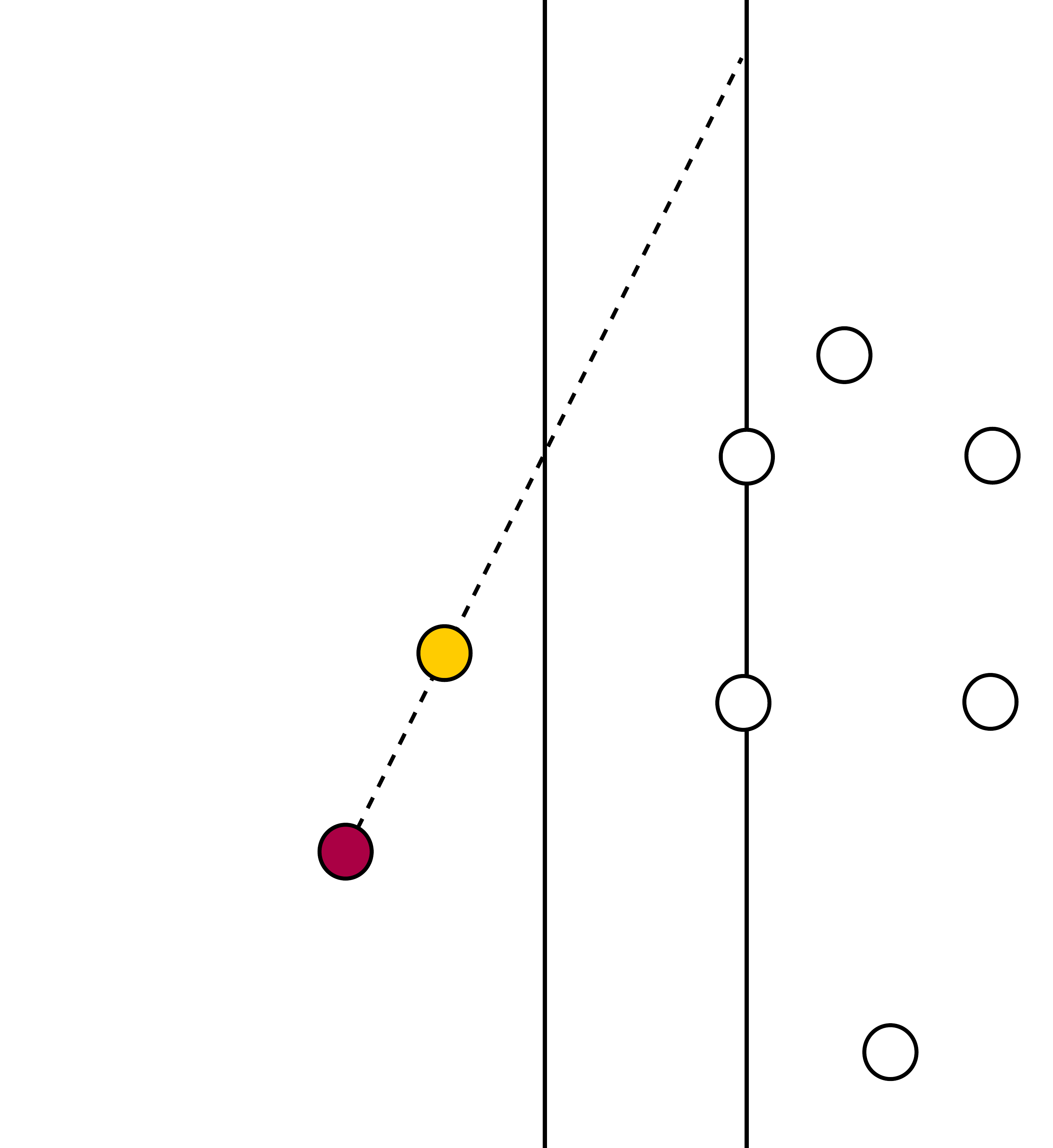
}

\caption[Short Caption]{Illustrations supporting the proof of Lemma \ref{lemma p1 hard}.}
\label{fig: proof A}
\end{figure}

\textbf{Case 2:} Let $t_2 \in (t'_1, t''_1)$. Clearly, all robots in $\mathcal{R}' \cup \mathcal{R}''$ are stable in $[0, t_2]$. So, when $r_2$ takes the snapshot at $t_2$, it finds $r_1$ in $\mathcal{H}_{L}^{O}(r_2)$ with light set to \texttt{terminal}. Hence, \textsc{LeftMostTerminal()} will return False. Therefore, all robots in $\mathcal{R} \setminus \{r_1\}$ will be stable in $(t'_1, t''_1)$. 
Hence, this case reduces to Case 3. 

\textbf{Case 3:} Let $t_2 \geq t''_1$. Similarly as shown in Case 1, all robots in $\mathcal{R}' \cup \mathcal{R}''$ are stable in $[0, t_2]$. Clearly, \textsc{LeftMostTerminal()} will return True for $r_2$. So, it will change its light to \texttt{terminal}, leave $\mathcal{L}$ (at some time $t'_2 > t_2$), move $d$ distance leftwards and change its light to \texttt{candidate} (at some time $T_1 > t'_2$). We only need to show that all robots in $\mathcal{R}' \cup \mathcal{R}''$ are stable in $(t'_2, T_1]$. Notice that $r_2 \in \mathcal{H}_{R}^{C}(r_1)$ at any time in the interval $(t'_2, T_1]$. Hence, any $r \in \mathcal{R}' \cup \mathcal{R}''$ that takes a snapshot at $t \in (t'_2, T_1)$ can either see $r_2$ with light \texttt{terminal} in $\mathcal{H}_{L}^{O}(r)$ or otherwise a robot (that lies on the line segment joining $r$ and $r_2$) with light set to \texttt{off} in $\mathcal{H}_{L}^{O}(r)$. In either case, \textsc{LeftMostTerminal()} will return False for $r$. Therefore, at time $T_1$, we obtain a candidate configuration $\mathbb{C}(T_1)$. 

\end{proof}

\begin{lemma}\label{lemma p1 line}
 If the initial configuration $\mathbb{C}(0)$ has all $n$ robots on the same vertical line, then $\exists~\ T_1 > 0$ such that $\mathbb{C}(T_1)$ is either a candidate configuration or a leader configuration.
\end{lemma}

\begin{proof}
 Suppose that all the robots in $\mathbb{C}(0)$ are on the vertical line $\mathcal{L}$. Let $r_1$ and $r_2$ be the terminal robots on $\mathcal{L}$. It is easy to see that eventually both $r_1$ and $r_2$ will start moving towards left, and meanwhile all other robots will remain stable. When both $r_1$ and $r_2$ have left $\mathcal{L}$, every $r \in \mathcal{R} \setminus \{r_1,r_2\}$ is in $\mathcal{H}^{O}(r_1, r_2)$ (see Fig. \ref{fig: proof B}). Hence, each $r \in \mathcal{R} \setminus \{r_1,r_2\}$ can see both $r_1$ and $r_2$ on its left. So, all robots in $\mathcal{R} \setminus \{r_1,r_2\}$ will remain stable. 
 
 After reaching their destinations, $r_1$ and $r_2$  will change their lights to \texttt{candidate}. The robot that changes its light first,  will remain stable until the other one changes its light to \texttt{candidate}. So, there is a time $t$ when $r_1.light = r_2.light = $ \texttt{candidate}, and all robots are stable.
 
 If $1_{r_1} = 1_{r_2}$, then at time $T_1 = t$, we have a candidate configuration. Now assume that $1_{r_1} \neq 1_{r_2}$, $1_{r_1} > 1_{r_2}$. So when $r_1$ observes $r_2$ with light \texttt{candidate} in $\mathcal{H}_{R}^{O}(r_1)$, it will change its light to \texttt{off}. When $r_2$ finds $r_1$ with light \texttt{off} in $\mathcal{H}_{L}^{O}(r_2)$, it will also change its light to \texttt{off}. Hence, at some time $t' > t$, $r_1$ will find that it is eligible to become leader. Then by Lemma \ref{lemma eligible to leader}, we have a leader configuration at some time $T_1 > t'$. 
 
\end{proof}

 \begin{figure}[!h]
\centering
\subcaptionbox[Short Subcaption]{
       \label{}
}
[
    0.33\textwidth 
]
{
    \fontsize{8pt}{8pt}\selectfont
    \def\svgwidth{0.33\textwidth}
    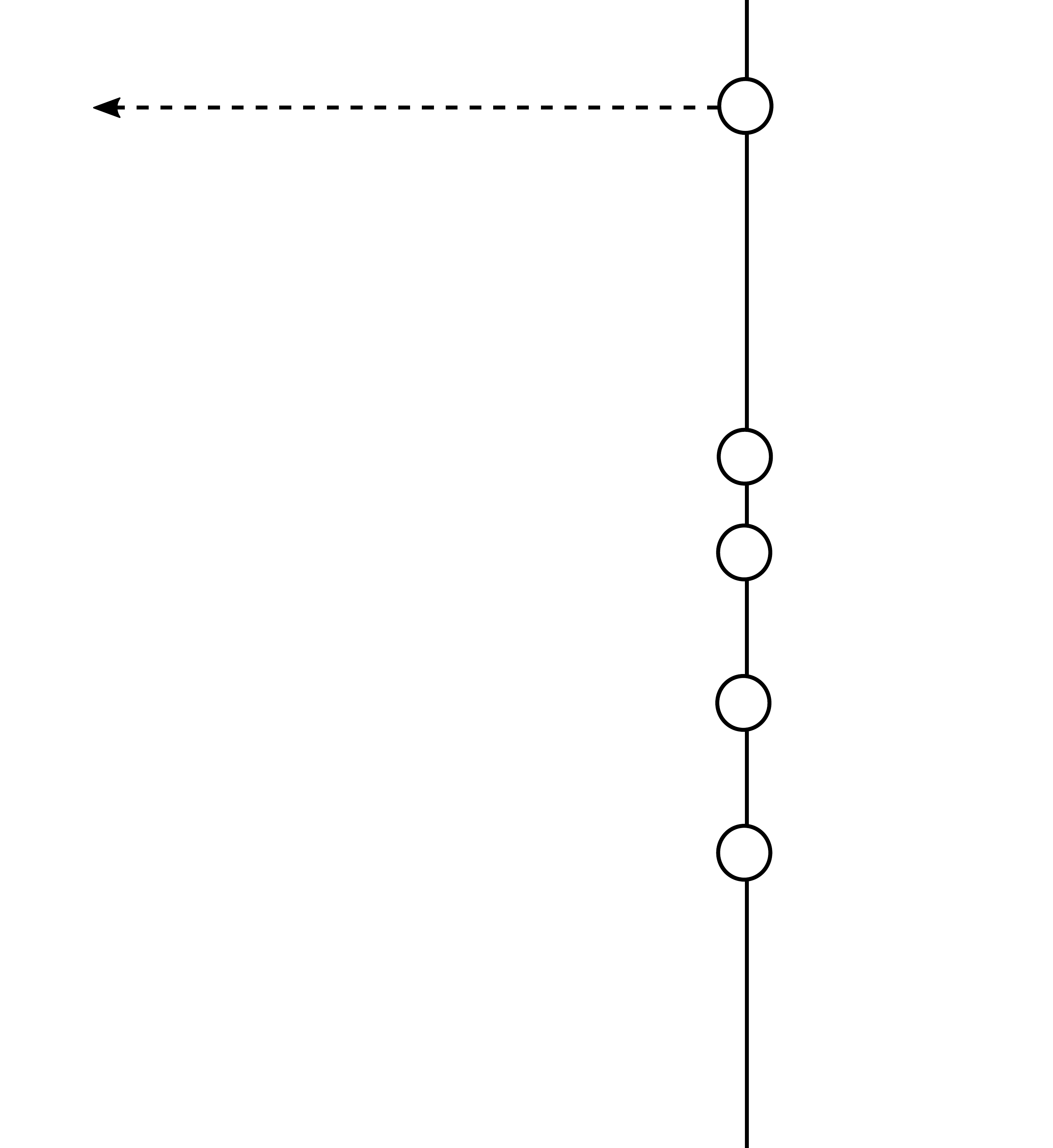
}
\hspace*{2cm}
\subcaptionbox[Short Subcaption]{
     \label{}
}
[
    0.33\textwidth 
]
{
    \fontsize{8pt}{8pt}\selectfont
    \def\svgwidth{0.33\textwidth}
    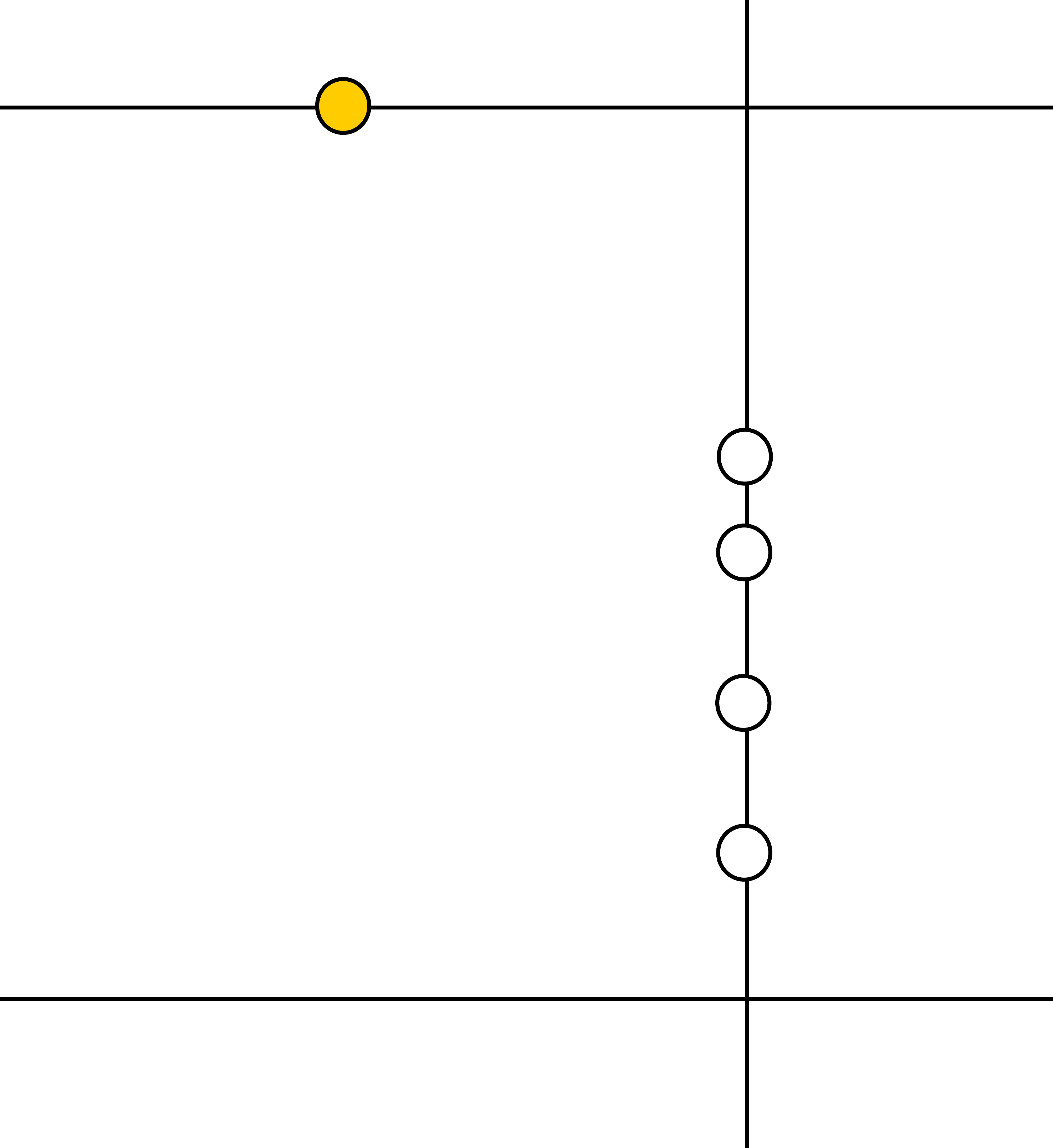
}

\caption[Short Caption]{Illustrations supporting the proof of Lemma \ref{lemma p1 line}.}
\label{fig: proof B}
\end{figure}

\begin{theorem}\label{thm: t1}
 For any initial configuration $\mathbb{C}(0)$, $\exists~\ T_1 > 0$ such that $\mathbb{C}(T_1)$ is a leader configuration or a candidate configuration.
\end{theorem}

\begin{proof}
 Follows from Lemma \ref{lemma p1 easy}, \ref{lemma p1 hard} and \ref{lemma p1 line}. 
\end{proof}

\begin{lemma}\label{lemma: phase 2a}
 Let $\mathbb{C}$ be a configuration of robots $\mathcal{R}$, where 
 \begin{enumerate}
  \item there are two robots $r_1$ and $r_2$ on the vertical line $\mathcal{L}$,
  
  \item either $r_1.light = r_2.light =$ \texttt{candidate}, or  $r_1.light \neq r_2.light$ with $r_1.light, r_2.light \in$ $\{\texttt{candidate},\\ \texttt{symmetry}\},$ 
  
  \item  all other robots are on the right of $\mathcal{L}$ (i.e., in  $\mathcal{H}_R^O(r_1) = \mathcal{H}_R^O(r_2)$) and have their lights set to \texttt{off}. 
 \end{enumerate}

 Then any robot $r \in \mathcal{R} \setminus \{r_1, r_2\}$ taking snapshot of this configuration will not move or change its color.
\end{lemma}

\begin{proof}
 Let $\mathcal{L'}$ be the leftmost vertical line containing a robot from $ \mathcal{R} \setminus \{r_1, r_2\}$. Then any robot on $\mathcal{L'}$ can see both $r_1, r_2$ on their left. For any robot in $ \mathcal{R} \setminus \{r_1, r_2\}$ and not lying on $\mathcal{L'}$, it will see at least one robot with light \texttt{off} on its left. Hence, any robot in $ \mathcal{R} \setminus \{r_1, r_2\}$ taking snapshot of this configuration will not move or change its color. 
\end{proof}

\begin{lemma}\label{lemma: phase 2b}
 Let $\mathbb{C}(t)$ be a configuration of robots $\mathcal{R}$, where 
 \begin{enumerate}
  \item there are two robots $r_1$ and $r_2$ on the vertical line $\mathcal{L}$ with their lights set to \texttt{candidate},
  
  \item  all other robots are on the right of $\mathcal{L}$ (i.e., in  $\mathcal{H}_R^O(r_1) = \mathcal{H}_R^O(r_2)$), have their lights set to \texttt{off} and stable. 
 \end{enumerate}

 If i) $r_2$ takes a snapshot at $t$ and moves left by the distance computed by \textsc{ComputeDestination2()}, and ii) during its movement,  $r_1$ is either stationary or moving away from $r_2$ along $\mathcal{L}$, then the robots in $\mathcal{R} \setminus \{r_1,r_2\}$ will remain stable during the movement of $r_2$.
\end{lemma}

\begin{proof}
 
 Let $\mathcal{H}$ be the open half-plane delimited by $\mathcal{L}_{H}(r_1)$ that does not contain $r_2$. Let $\mathcal{H}^c$ be the complement of  $\mathcal{H}$. Let $\mathcal{L}'$ be the leftmost vertical line containing a robot from $\mathcal{R} \setminus \{r_1, r_2\}$. It is sufficient to show that all robots on $\mathcal{L'}$ will be able to see both $r_1$ and $r_2$ during the movement of $r_2$. It is easy to see that any robot on $\mathcal{L}' \cap \mathcal{H}^c$ will be able to see both $r_1$ and $r_2$ during the movement of $r_2$. So, assume that $\mathcal{L}' \cap \mathcal{H}$ has at least one robot. Let  $r_3$ be the robot on $\mathcal{L}' \cap \mathcal{H}$ with maximum $y$-coordinate (in the local coordinate system of $r_2$ as set in line \ref{code: set y axis} of Algorithm \ref{algo:phase2}). Consider the line $\mathcal{M}$ joining $r_3(t)$ and $r_1(t)$ (see Fig. \ref{fig: proof C1}). It is easy to see that \textsc{ComputeDestination2()} ensures that $r_2$ does not cross the line $\mathcal{M}$. This implies that all robots on $\mathcal{L}' \cap \mathcal{H}$ will always be able to see both $r_1$ and $r_2$.

\end{proof}

 \begin{figure}[h]
\centering
%
{
    \fontsize{8pt}{8pt}\selectfont
    \def\svgwidth{0.45\textwidth}
    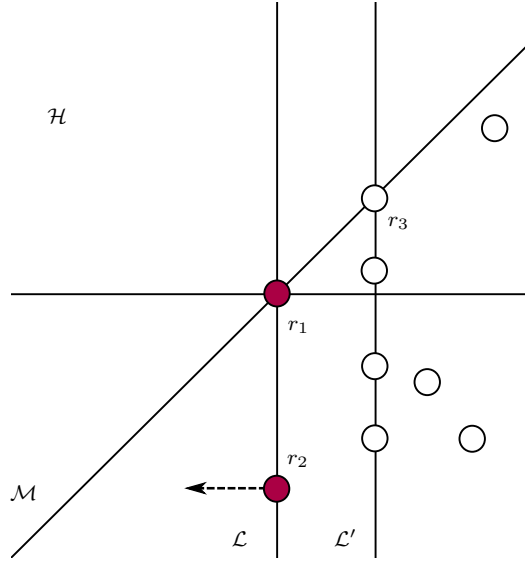
}

\caption[Short Caption]{Illustrations supporting the proof of Lemma \ref{lemma: phase 2b}.}
\label{fig: proof C1}
\end{figure}

\begin{lemma}\label{lemma: phase 2c}
 Let $\mathbb{C}$ be a configuration of $n$ robots $\mathcal{R}$, where 
 \begin{enumerate}
  \item there are two robots $r_1$ and $r_2$ on the vertical line $\mathcal{L}$,
  
  \item  all other robots are on the right of $\mathcal{L}$ (i.e., in  $\mathcal{H}_R^O(r_1) = \mathcal{H}_R^O(r_2)$).
 \end{enumerate}

 If $r_1$ moves along $\mathcal{L}$ in the direction opposite to $r_2$ and all other robots remain stationary, then after finite number of steps, $r_1$ can see all the robots.
\end{lemma}

\begin{proof}
 Take all possible pairs of robots from $\mathcal{R} \setminus \{r_1,r_2\}$ and draw the straight lines joining them. Consider the points where these lines intersect $\mathcal{L}$. Clearly, the robot $r_1$ will not be able to see all the other robots in the configuration if and only if it is positioned on one of these points. Since there are at most $^{n-2}C_2$ of these points, after finitely many steps along $\mathcal{L}$ in one direction, $r_1$ will be able to see all the robots. 
\end{proof}

\begin{theorem}\label{thm: t1t2}
 If $\mathbb{C}(T_1)$ is a candidate configuration, then $\exists~\ T_2 > T_1$ such that $\mathbb{C}(T_2)$ is a leader configuration.
\end{theorem}

\begin{proof}
 
 Let $r_1$ and $r_2$ be the two robots with light set to \texttt{candidate}. Let $\mathcal{L}$, $\mathcal{L}'$, $\mathcal{K}$ and $\mathbb{C}'$ be as defined in Section \ref{sec: p2}. 
 
 First consider the case where $r_1$ and $r_2$ are in the same closed half-plane delimited by $\mathcal{K}$. As long as they are on $\mathcal{L}$, all other robots will remain stable by Lemma \ref{lemma: phase 2a}. Let $d_{r_2} > d_{r_1}$. Then $r_1$ will remain stationary, and $r_2$ will move left by the distance computed by \textsc{ComputeDestination2()}. When $r_2$ leaves $\mathcal{L}$, we are back to phase 1.
 By Lemma \ref{lemma: phase 2b}, all robots in $ \mathcal{R} \setminus \{r_2\}$ will remain stable during the movement of $r_2$. Now in the next snapshot, $r_2$ will find $r_1$ in right with light set to \texttt{candidate} and hence, it will change its light to \texttt{off}. Then $r_1$ will also change its light to \texttt{off}. So $r_2$ will find itself eligible to become leader, and hence by Lemma \ref{lemma eligible to leader}, we shall obtain a leader configuration at time $T_2 > T_1$.

Now consider the case where $r_1$ and $r_2$ are not in the same closed half-plane delimited by $\mathcal{K}$. First consider the case where  $\lambda'(r_1) \neq \lambda'(r_2)$, say $\lambda'(r_1) \prec \lambda'(r_2)$. Both $r_1$ and $r_2$ can determine this as they can see all the robots on $\mathcal{L}'$. Therefore, $r_1$ will decide to move left. Again by Lemma \ref{lemma: phase 2b}, all robots in $ \mathcal{R} \setminus \{r_1\}$ will remain stable during the movement of $r_1$. As before, we shall obtain a leader configuration at time $T_2 > T_1$. So now consider the cases when $\lambda'(r_1) = \lambda'(r_2)$.
 
 \textbf{Case 1:} $\mathbb{C}'$ is symmetric with respect to $\mathcal{K}$ and there are no robots on it.  We claim that $d_{r_1} \neq d_{r_2}$. If $d_{r_1} = d_{r_2}$, then $\mathbb{C}(T_1)$ is symmetric with respect to $\mathcal{K}$ which contains no robots.  Clearly, $r_1$ and $r_2$ were two terminal leftmost robots in the initial configuration $\mathbb{C}(0)$. The candidate configuration $\mathbb{C}(T_1)$ is created when $r_1$ and $r_2$ have moved left by equal distances. Hence, it implies that $\mathbb{C}(0)$ was also symmetric with respect to $\mathcal{K}$ having no robots on $\mathcal{K}$. This contradicts our assumption regarding the initial configuration. So, assume that $d_{r_1} < d_{r_2}$. Then $r_1$ will remain stable and $r_2$ will move along $\mathcal{L}$ until the condition in line \ref{code: condition p2} of Algorithm \ref{algo:phase2} is satisfied.  It will then find that $\lambda(r_1) = \lambda(r_2)$ and $\mathcal{K}$ contains no robots. Then $r_2$ will decide to move left and as before, we shall obtain a leader configuration at time  $T_2 > T_1$.
 
 \textbf{Case 2:} $\mathbb{C}'$ is asymmetric with respect to $\mathcal{K}$. First assume that $d_{r_1} \neq d_{r_2}$, say $d_{r_1} < d_{r_2}$. Then $r_1$ will remain stable and $r_2$ will move along $\mathcal{L}$ until the condition in line \ref{code: condition p2} of Algorithm \ref{algo:phase2} is satisfied. If $\lambda(r_1) \succ \lambda(r_2)$, $r_2$ will decide to move left and as before, we shall obtain a leader configuration at time  $T_2 > T_1$. If $\lambda(r_1) \prec \lambda(r_2)$, $r_2$ will move right. Since all robots lie inside $\mathcal{H}_U^C(r_2)$, all robots in $\mathcal{L'}$ will be able to see both $r_1$ and $r_2$ during the movement of $r_2$. Similarly as before, we shall obtain a leader configuration at time $T_2 > T_1$ with $r_1$ as the leader. Now consider the remaining case when $d_{r_1} = d_{r_2}$ at time $t$. The adversary may activate $r_1$ and $r_2$ in such a manner (e.g. synchronously) so that both $r_1$ and $r_2$ move along $\mathcal{L}$ and they both decide to leave $\mathcal{L}$.  In that case, suppose $r_1$ moves  right and $r_2$ moves  left. Note that when they decided to move, the condition in line \ref{code: condition p2} of Algorithm \ref{algo:phase2} is satisfied for both of them. Hence, all the robots are inside $\mathcal{H}^O(r_1, r_2)$ during the movements of $r_1$ and $r_2$. Therefore all robots on $\mathcal{L'}$ will be able to see $r_1$ and $r_2$ during their movements. As before, we shall obtain a leader configuration at time  $T_2 > T_1$.
 
 \textbf{Case 3:} Let $\mathbb{C}'$ be symmetric with respect to $\mathcal{K}$ and there is at least one robot on it. We shall show that 1) both $r_1$ and $r_2$ will set their lights to \texttt{symmetry}, and 2) when they do so, they must be equidistant from $\mathcal{K}$.

  \textbf{Case 3A:} First assume that $d_{r_1} \neq d_{r_2}$, say $d_{r_1} < d_{r_2}$. Then $r_1$ will remain stable and $r_2$ will move along $\mathcal{L}$ until the condition in line \ref{code: condition p2} of Algorithm \ref{algo:phase2} is satisfied. Then $r_2$ changes its light to \texttt{symmetry}. When $r_1$ sees $r_2$ with light \texttt{symmetry}, $r_1$ will move $d_{r_2}- d_{r_1}$ distance away from $r_2$. In the next LCM cycle, $r_1$ will change its light to \texttt{symmetry}, say at time $t$. By Lemma \ref{lemma: phase 2a}, all robots in $\mathbb{C}'$ remain stable in $[T_1,t]$.

  \textbf{Case 3B:} Next assume that initially $d_{r_1} = d_{r_2}= D$. Therefore, the adversary can make both $r_1$ and $r_2$ move. By Lemma \ref{lemma: phase 2a},  all robots in $\mathbb{C}'$ will remain stable until they both change their lights to \texttt{symmetry}.  Let $\mathcal{H}_1$ and $\mathcal{H}_2$ be the open half-planes delimited by $\mathcal{K}$ such that $r_1 \in \mathcal{H}_1$ and $r_2 \in \mathcal{H}_2$. For $i = 1, 2,$ let $p_{i,0}, p_{i,1}, p_{i,2}, \ldots$ be the points on $\mathcal{L} \cap \mathcal{H}_i$ at distances $D, 2^{1}D, 2^{2}D, \ldots$ from $\mathcal{K}$. Clearly, if $r_i$ moves according to Algorithm \ref{algo:phase2}, it stops at these points and take snapshots.  Let $k \geq 0$ be the smallest non-negative integer such that when $r_1$ reaches $p_{1,k}$, it will find that the conditions in line \ref{code: condition p2} of Algorithm \ref{algo:phase2} are satisfied. Due to the symmetry of $\mathbb{C}'$, $k$ is also the smallest non-negative integer such that $r_2$ will find the conditions in line \ref{code: condition p2} of Algorithm \ref{algo:phase2} to be satisfied at $p_{2,k}$. If $k=0$, there is nothing to prove as $r_1$ and $r_2$ will change their colors to \textit{symmetry} at their starting positions $p_{1,0}$ and $p_{2,0}$ respectively. So, let $k > 0$. Notice, that if any one of $r_i$, say $r_1$, takes snapshot at $p_{1,j}$, $j<k$ and finds $r_2$ at a distance greater than $2^{j}D$ from $\mathcal{K}$, then the situation reduces to case 3A. Otherwise, $r_1$ and $r_2$ will reach $p_{1,k}$ and $p_{2,k}$ respectively and change their colors to \texttt{symmetry}.


 Therefore, we have shown that at some time $t > T_1$, $r_1$ and $r_2$ are equidistant from $\mathcal{K}$ and have their lights set to \texttt{symmetry}. Suppose $r_3$ is the leftmost robot on $\mathcal{K}$. When $r_3$ sees the two robots with light \texttt{symmetry} in $\mathcal{H}_L^O(r_3)$, it will move to the left of $\mathcal{L}$. When $r_1$ and $r_2$ find $r_3$ on their left, they will change their lights to \texttt{off}. Hence, eventually $r_3$ will find itself eligible to become leader, and then by Lemma \ref{lemma eligible to leader}, we shall obtain a leader configuration at time $T_2 > t > T_1$.

   \end{proof}

\begin{theorem}
 For any initial configuration $\mathbb{C}(0)$, $\exists~\ T_2 > 0$ such that $\mathbb{C}(T_2)$ is a leader configuration.
\end{theorem}
 
 \begin{proof}
  By Theorem \ref{thm: t1},  $\exists~\ T_1 > 0$ such that $\mathbb{C}(T_1)$ is a leader configuration or candidate configuration. If $\mathbb{C}(T_1)$ is a leader configuration, then $T_2 = T_1$. By Theorem \ref{thm: t1t2}, if $\mathbb{C}(T_1)$ is a candidate configuration, then $\exists~\ T_2 > T_1$ such that $\mathbb{C}(T_2)$ is a leader configuration.  
 \end{proof}

A stable configuration of $n$ robots is called an \emph{agreement configuration} if there are two robots $r_l$ and $r_u$ such that
\begin{enumerate}
 \item $r_l.light =$ \texttt{leader} and $r.light =$ \texttt{off} for all $r \in \mathcal{R} \setminus \{r_l\}$
 
 \item  $r \in \mathcal{H}_U^O(r_l)$ for all $r \in \mathcal{R} \setminus \{r_l\}$
 
 \item there is exactly one robot $r_u \in \mathcal{R} \setminus \{r_l\}$ that is on the same vertical line with $r_l$
 
 \item $r \in \mathcal{H}_U^C(r_u)$ for all $r \in \mathcal{R} \setminus \{r_l, r_u\}$. 
\end{enumerate}

In the proofs of the following three lemmas, $r_1, r_2, \ldots, r_{n-1}$ and  $s_0, s_1, \ldots, s_{n-1}$ will denote the non-leader robots and the target points, as described in Section \ref{sec:election} (see Fig. \ref{fig: stage2_robot_order} and \ref{fig: stage2_target0}).

\begin{lemma}\label{lemma: leader to agreement}
 There exists $T_3 > T_2$ such that $\mathbb{C}(T_3)$ is an agreement configuration.
\end{lemma}

\begin{proof}

 Clearly, $r_1$ will move horizontally to place itself on the line $\mathcal{L}$ (line \ref{code: condition stage2} and \ref{code: unit condition} of Algorithm \ref{main_algorithm: stage 2}). We show that during the movement of $r_1$, all other robots will be stable. Consider any robot $r \in \mathcal{R} \setminus \{r_l,r_1\}$. 
 If it can see $r_l$, it decides that it is in stage 2. It will do nothing as clearly the condition in line \ref{code: condition stage2} of Algorithm \ref{main_algorithm: stage 2} is not satisfied. If it can not see $r_l$, it decides that it is in stage 1. It is easy to see by inspecting the algorithm for stage 1, that it will remain stable. Hence, we have an agreement configuration at  some $T_3 > T_2$. 
\end{proof}


A stable configuration of $n$ robots is called an \emph{L-configuration} if there are two robots $r_l$ and $r_u$ such that
\begin{enumerate}
 \item $r_l.light =$ \texttt{leader} and $r.light =$ \texttt{off} for all $r \in \mathcal{R} \setminus \{r_l\}$
 
 \item  there is exactly one robot $r_u \in \mathcal{R} \setminus \{r_l\}$ that is on the same vertical line with $r_l$

 \item all robots in $\mathcal{R} \setminus \{r_l, r_u\}$ are on $\mathcal{L}_H(r_l) \cap \mathcal{H}_R^O(r_l)$ at points $(\Psi(2), -1), (\Psi(3), -1),\\ \ldots, (\Psi(n-1), -1)$ in the agreed coordinate system. 
\end{enumerate}

\begin{lemma}\label{lemma: agreement to L}
 There exists $T_4 > T_3$ such that $\mathbb{C}(T_4)$ is an L-configuration.
\end{lemma}


\begin{proof}

  Clearly, $r_2$ will first decide to move to $(\Psi(2), -1)$. The condition in line \ref{code: condition stage2} of Algorithm \ref{main_algorithm: stage 2} guarantees that for $i > 2$, $r_i$ will start moving after $r_2, r_3, \ldots, r_{i-1}$ have completed their moves. When this condition is satisfied for $r_i$, it can see both $r_l$ and $r_u$, and hence can determine the point with coordinates $(\Psi(2), -1), (\Psi(3), -1), \ldots, (\Psi(n-1), -1)$. Clearly, $r_i$ will find $i-1$ robots on $\mathcal{L}_H(r_l)$ and will decide to move to $(\Psi(i), -1)$. Again, it is easy to see that all other robots will be stable during its movement. Therefore at some time  $T_4 > T_3$, an L-configuration is created. 
\end{proof}

\begin{lemma}\label{lemma: L to pattern}
 There exists $T_5 > T_4$ such that $\mathbb{C}(T_5)$ is a final configuration similar to the given pattern and has all robots with light set to \texttt{done}.
\end{lemma}

\begin{proof}

 First we show that the robots $r_2, r_3, \ldots, r_{n-1}$ will sequentially move to  $s_2, s_3, \ldots, s_{n-1}$. When $r_2$ takes the first snapshot, it finds the condition in line \ref{code: final condition1} of Algorithm \ref{main_algorithm: stage 2} is satisfied. So, it changes its light to \texttt{done} and moves to $s_2$. For $i > 2$, $r_i$ clearly does not move until $r_{i-1}$ reaches $s_{i-1}$ (see line \ref{code: final condition2} of Algorithm \ref{main_algorithm: stage 2}). We show that after $r_{i-1}$ reaches $s_{i-1}$, the conditions in line \ref{code: final condition1} and line \ref{code: final condition1b} of Algorithm \ref{main_algorithm: stage 2} are satisfied for $r_i$. This is because of the following.
 \begin{enumerate}
 
  \item $r_i$ can see $r_l$ on $\mathcal{L}_H(r_i)$ as $r_2, \ldots, r_{i-1}$ have already moved.

  \item $r_i$ can see $r_u$ as all $r_2, \ldots, r_{i-1}$ are in $\mathcal{H}^C_U(r_u)$ (because all entries of $\mathbb{P}$ are from $\mathbb{R}^2_{\geq 0}$).
    
  \item Obviously there are no robots with light \texttt{off} in $\mathcal{H}_U^O(r_i) \cap \mathcal{H}_R^O(r_l)$.
  
  \item Since $r_i$ can see both $r_l$ and $r_u$, it can compute its position in the global coordinate system, and should find it to be $(\Psi(i), -1)$.
  
  \item It is easy to see the line segment joining $s_{i-1}$ and $(\Psi(i), -1)$ does not pass through any of $s_2, \ldots, s_{i-2}$. Hence, $r_i$ can see  $r_{i-1}$ at $s_{i-1}$ as it is not obstructed by $r_2, \ldots, r_{i-2}$. Hence, it will find $r_{i-1}$ with light \texttt{done} at $s_{i-1}$.
  
\end{enumerate}

 So, $r_i$ will change its light to \texttt{done} and move to $s_i$. Therefore, the robots $r_2, r_3, \ldots, r_{n-1}$ will sequentially move to  $s_2, s_3, \ldots, s_{n-1}$. Then $r_u$ will find that the condition in line \ref{code: final condition3} of Algorithm \ref{main_algorithm: stage 2} is satisfied and hence, will move to $s_1$ (without changing its light). In the next LCM-cycle, it will change its color to \texttt{done} (see line \ref{code: unit off to done} in Algorithm \ref{main_algorithm: stage 2}). When $r_l$ sees no robots with light set to \texttt{off}, it changes it light to \texttt{done} and moves to $s_0$. It is clear from the discussions at the end of Section \ref{sec:election} that $r_l$ is able to determine the point $s_{0}$ in the plane. Hence, the given pattern is formed at some time $T_5 > T_4$ with all the robots having their lights set to \texttt{done}. The configuration $\mathbb{C}(T_5)$ is final as robots with color \texttt{done} do not do anything. 
 \end{proof}
 
 \begin{theorem}
  For the leader configuration $\mathbb{C}(T_2)$, there exists $T_5 > T_2$ such that $\mathbb{C}(T_5)$ is a final configuration similar to the given pattern and has all robots with light set to \texttt{done}.
 \end{theorem}

 \begin{proof}
 Follows from Lemma \ref{lemma: leader to agreement}, \ref{lemma: agreement to L} and \ref{lemma: L to pattern}. 
 \end{proof}

\subsection{The Main Results}
From the results proved in the last subsection, we can conclude the following result.

\begin{theorem}
 For a set of opaque and luminous robots with one axis agreement, $\mathcal{APF}$ is deterministically solvable if and only if the initial configuration is not symmetric with respect to a line $\mathcal{K}$ such that 1) $\mathcal{K}$ is parallel to the agreed axis and 2) $\mathcal{K}$ is not passing through any robot. Six colors are sufficient to solve the problem from any solvable initial configuration.
 \end{theorem}

 If the robots agree on the direction and orientation of both axes, then leader election is easy. If there is a unique leftmost robot $r$, then as before it will become leader (by executing \textsc{BecomeLeader()}). If there are multiple leftmost robots, then the bottommost one will move left. In the next snapshot, it will find itself eligible to become leader and start executing \textsc{BecomeLeader()}. Therefore, leader election is solvable using only 2 colors, namely \texttt{off} and \texttt{leader}. Then stage 2 will be executed by Algorithm \ref{main_algorithm: stage 2}. Hence, we have the following result. 
 
\begin{theorem}
 For a set of opaque and luminous robots with two axis agreement, $\mathcal{APF}$ is deterministically solvable from any initial configuration using three colors.
 \end{theorem}

\section{Concluding Remarks}\label{sec:conclusion}
In this work, we have provided a full characterization of initial configurations from where $\mathcal{APF}$ is deterministically solvable in \textsc{ASync} for a system of opaque and luminous robots with one and two axis agreement. Our long term goal is to obtain a characterization of initial configurations from where $\mathcal{APF}$ is deterministically solvable in the obstructed visibility model without any agreement in coordinate system. This paper is the first step towards it. The next step could be to consider robots with only  chirality. Our algorithms require 3 and 6 colors respectively for two axis and one axis agreement. We have also shown that $\mathcal{APF}$ is deterministically unsolvable by oblivious robots even in two axis agreement. This means that we need at least 2 colors to solve the problem in both two axis and one axis agreement. An interesting open problem is to find the minimum number of colors to solve the problem.

 \begin{figure}[thb!]
\centering
\subcaptionbox[Short Subcaption]{
       \label{}
}
[
    0.35\textwidth 
]
{
    \fontsize{8pt}{8pt}\selectfont
    \def\svgwidth{0.35\textwidth}
    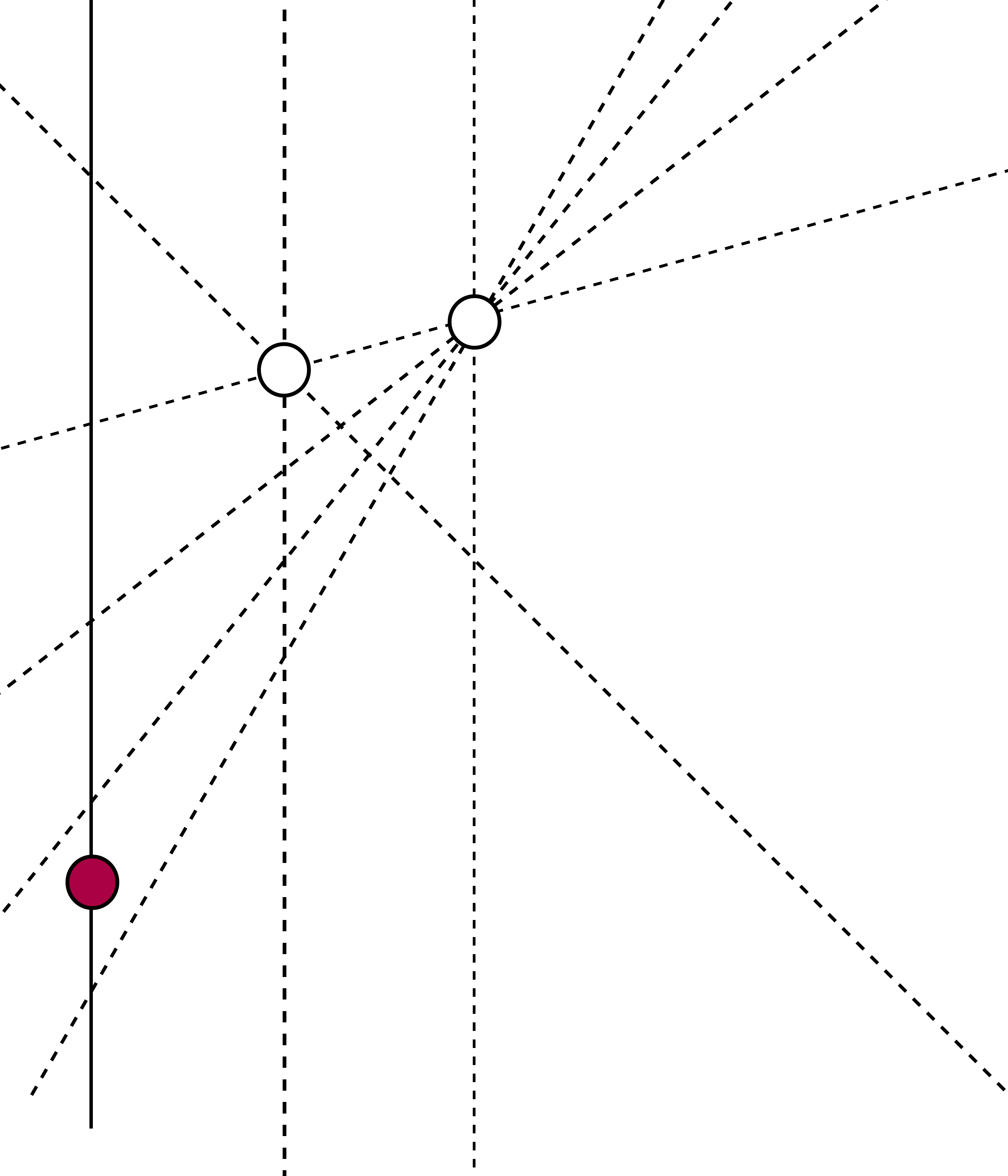
}
\hspace*{1cm}
\subcaptionbox[Short Subcaption]{
     \label{}
}
[
    0.35\textwidth 
]
{
    \fontsize{8pt}{8pt}\selectfont
    \def\svgwidth{0.35\textwidth}
    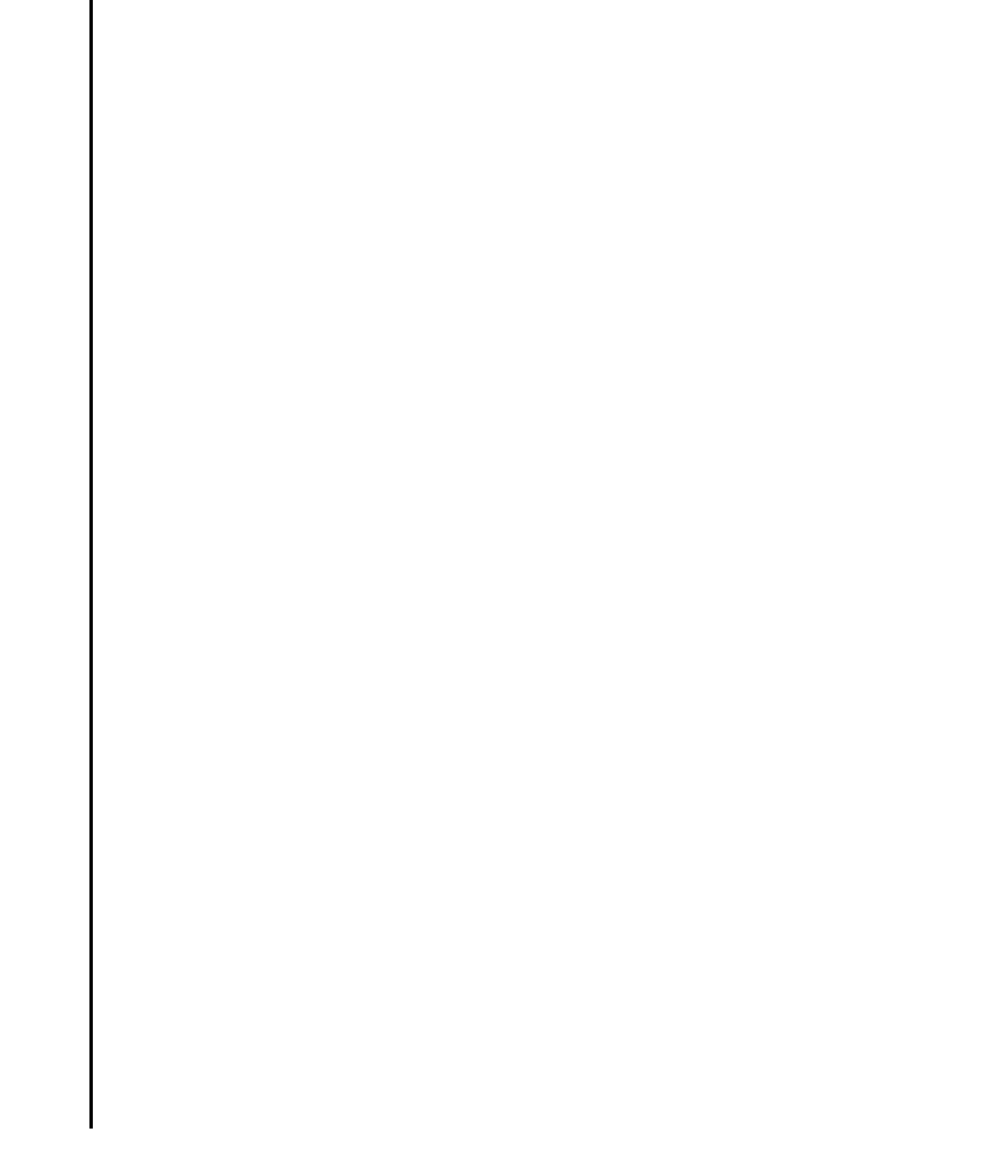
}

\caption[Short Caption]{a) A configuration of point robots. The small robots are to be understood as dimensionless points. There are finitely many points on $\mathcal{L}$ from where the full configuration can not be seen. b) A configuration of fat robots. There is no point on $\mathcal{L}$ from where the robot $r$ can be seen. In fact, unless the robots guarding $r$ move, it is impossible for $r_1, r_2$ to see $r$ no matter where they move to.}
\label{fig: fat}
\end{figure}

The robot model considered in this work assumes that the robots are dimensionless points in the plane. This is not a realistic assumption as in reality, even very small robots occupy some space. Therefore, a more realistic model is the \emph{opaque fat robots} model where the robots are modeled as disks. In this case, two robots $r$ and $r'$ will be visible to each other if and only if there exist points $p$ and $p'$ on the boundary of $r$ and $r'$ respectively such that the line segment joining $p$ and $p'$ does not pass through any other point of any robot. It is not very difficult to modify Stage 2 of our algorithm to adapt it to the opaque fat robot setting without requiring any extra colors. Therefore, $\mathcal{APF}$ is solvable in two axis agreement using 3 colors as \textsc{Leader Election} is trivial in this case. However,  this is not the case in the one axis agreement setting. Recall that our \textsc{Leader Election} strategy was to move the two candidate robots so that they can get the full view of the configuration. This strategy worked because, for the two candidate robots moving along the vertical line $\mathcal{L}$, there are only finitely many points on $\mathcal{L}$ from where the full configuration can not be seen. This is essentially due to the fact that the robots were dimensionless points: a candidate robot can not see the full configuration if and only if it is on a point on $\mathcal{L}$ that is collinear with two non-candidate robots, and there can be at most $^{n-2}C_2$ such points. This does not work for fat robots as there could be no such point on $\mathcal{L}$ from where all robots are visible (see Fig. \ref{fig: fat}). Recently in \cite{BoseAKS20}, $\mathcal{APF}$ has been solved for opaque fat robots with one axis agreement. In \cite{BoseAKS20}, instead of moving the robots so that they can see the entire configuration, some geometric information are communicated among the robots using lights, which are used to elect the leader. In other words, movements are traded off for communication. As a consequence the total number of moves executed by all the robots is $\Theta(n)$. This is also asymptotically optimal. To see this, consider an initial configuration where all robots are collinear. Then if the pattern to be formed does not have three collinear robots, then at least $n-2$ robots need to move. So, the total number of moves required to solve $\mathcal{APF}$ is $\Omega(n)$. However, our algorithm could require $O(n^2)$ moves as in Phase 2, the candidate robots could move $^{n-2}C_2 = O(n^2)$ times. Furthermore, the algorithm of \cite{BoseAKS20} works even for robots with non-rigid movement. All these improvements are achieved by trading off movements for communication using lights. The obvious drawback of this strategy is that it requires significantly more colors. It would be interesting to see if these improvements can be achieved using fewer colors. 


\paragraph{Acknowledgements.} The first and the third author are supported by NBHM, DAE, Govt. of India and CSIR, Govt. of India, respectively. This work was done when the second author was at Jadavpur University, Kolkata, India, supported by UGC, Govt. of India. We would like to thank the anonymous reviewers for their valuable comments which helped us to improve the quality and presentation of this paper.

\bibliographystyle{splncs04}
\bibliography{pattern_opaque}

\end{document}